\documentclass[preprint,12pt]{elsarticle}

\usepackage{stmaryrd}
\usepackage{graphicx}
\usepackage{amsthm}
\usepackage{subfig}
\usepackage{amsmath}
\usepackage{algorithm}
\usepackage{algorithmic}
\usepackage{nicefrac}
\usepackage{multirow}
\usepackage{paralist}
\usepackage{enumerate}
\usepackage{url}
\usepackage[all]{xy}
\usepackage{tikz}
\usepackage[utf8x]{inputenc}

\usepackage{amssymb}

\journal{Information and Computation}

\def\De{\Delta}

\def\mapping{\pi}

\def\mapping{\mathcal{PA}}
\def\WW{\mathcal{W}}
\def\BR{\mathbb{R}}
\def\somesymbol{\ast}
\def\TANGSTATES{S_{\mathit{nn-tangible}}}
\def\newsymbol{\mathfrak{C}}

\newcommand{\Dist}[1]{\mathrm{Dist}(#1)}
\newcommand{\SubDist}[1]{\mathrm{SubDist}(#1)}
\newcommand{\transitionRelationSymbol}{{T}}
\newcommand{\aut}[1][P]{{#1}}
\newcommand{\transitionRelationI}{\transitionRelation}
\newcommand{\transitionRelation}{\mathit{\transitionRelationSymbol}}
\newcommand{\stateSet}{S}
\newcommand{\startState}[1][s]{{#1}_0}
\newcommand{\actionSet}{Act}

\newcommand\aPA[1][]{\ensuremath{\aut#1 = (\stateSet#1, \actionSet#1, \transitionRelationI#1, \emptyset, \startState#1)}}

\usetikzlibrary{calc}
\pgfarrowsdeclaretriple{<<<}{>>>}{latex}{latex}%

\newcommand{\mapr}[1]
{\mathrel{\raisebox{0.38ex}{\begin{tikzpicture}[>=stealth]
\node[ inner sep=0pt,minimum width=5pt] (textbox) at (0,0)
{\scriptsize\ensuremath{#1}};
\coordinate (korrektur) at (0,-1pt);
\coordinate (end) at ($(textbox.south east)+(6pt,0) + (korrektur)$);
\draw ($(textbox.south east) + (korrektur)$) edge[-latex] (end);
\draw ($(textbox.south west)+(-4pt,0) + (korrektur)$) -- ($(textbox.south east) + (korrektur) $);
\fill ($(textbox.south east)+(-1pt,0) + (korrektur)$) circle (1pt);
\end{tikzpicture}}}
}

\newcommand{\mati}[1]
{\mathrel{\raisebox{0.55ex}{\begin{tikzpicture}[>=stealth]
\node[ inner sep=0pt,minimum width=5pt] (textbox) at (0,0)
{\scriptsize\ensuremath{#1}};
\coordinate (korrektur) at (0,-1pt);
\coordinate (end) at ($(textbox.south east)+(6pt,0) + (korrektur)$);
\draw ($(textbox.south east) + (korrektur)$) edge[->>>] (end);
\draw ($(textbox.south west)+(-4pt,0) + (korrektur)$) -- ($(textbox.south east) + (korrektur) $);
\end{tikzpicture}}}
} 

\newtheorem{lemma}{Lemma}
\newtheorem{example}{Example}
\newtheorem{remark}{Remark}
\newtheorem{corollary}{Corollary}

\newtheorem{definition}{Definition}

\newtheorem{theorem}{Theorem}

\newcommand{\transitionReplacement}[2]{#1(#2)}
\newcommand{\replace}[2]{#1{[#2]}}
\newcommand{\replacement}[2]{#2/#1}

\newcommand{\sd}{\nu} 

\newcommand{\probeval}[2]{#1(#2)}

\newcommand{\transitionsCombinedFromState}[1]{\newsymbol(#1)}

\newcommand{\setcond}[2]{\{\, #1 \mid #2 \,\}}
\newcommand{\setnocond}[1]{\{#1\}}

\newcommand{\hidden}{\tau}

\newcommand{\weakCombinedTransition}[3]{#1 \mystackrel{#2}{\Longrightarrow}_{\combined} #3}
\newcommand{\strongTransition}[3]{#1 \mystackrel{#2}{\longrightarrow} #3}

\newcommand{\combined}{\mathrm{C}}

\newcommand{\mystackrel}[2]{\stackrel{#1}{#2}}

\newcommand{\locallyChangedAut}[3][P]{{#1}_{(#2,#3)}}

\newcommand{\Supp}[1]{\mathrm{Supp}(#1)}

\newcommand{\weakBisimD}{\approx_\Delta}

\newboolean{useComments}
\setboolean{useComments}{true}
\ifthenelse{\boolean{useComments}}{%
        \newcommand{\todo}[1]{\textcolor{olive}{ TODO: {#1}}}
	\newcommand{\ms}[1]{{\color{red}\texttt{ M S: #1 :M S }}}
	\newcommand{\js}[1]{{\color{orange}\texttt{ J S: #1 :J S }}}
}{%
	\newcommand{\todo}[1]{}
	\newcommand{\ms}[1]{}
	\newcommand{\js}[1]{}
}

\graphicspath{{.}}

\begin{document}

\begin{frontmatter}



\title{Markov Automata: Deciding weak bisimulation by means of non-na\"ively vanishing states}

\author[unibw]{Johann Schuster}
\author[unibw]{Markus Siegle}
\address[unibw]{University of the Federal Armed Forces Munich}



\begin{abstract}
This paper develops a decision algorithm for weak bisimulation
on Markov Automata (MA).
For this purpose, different notions of vanishing states (a concept
originating from the area of Generalised Stochastic Petri Nets) are defined.
In particular, 
non-na\"ively vanishing states are shown to be essential for relating the concepts
of (state-based) na\"{i}ve weak bisimulation and (distribution-based) weak bisimulation.
The bisimulation algorithm presented here follows the partition-refinement scheme and has exponential time complexity.
\end{abstract}


\end{frontmatter}




\section{Introduction}

Markov Automata (MA) are a powerful formalism for modelling systems with nondeterminism, probability and continuous time. 
The weak bisimulation relation for MA \cite{lics:10,avacsreport} is not a relation
on the set of states, but rather a relation on the set of subdistributions
over states.
This is the reason, why it is not obvious how to develop an algorithm
for deciding distribution-based weak bisimulation for MA, and this is exactly the topic
of the present paper.

Our approach carries over some intuition from the area of Generalised Stochastic Petri Net \cite{marsan:95} to the MA setting. There, vanishing markings
are eliminated in order to minimise the number of reachable markings and to enable the subsequent steps of numerical analysis.
A basic example of a GSPN is given in Fig.~\ref{fig:GSPN_example}. It consists of the places $p_1$ to $p_4$,
an exponentially distributed transition $t_1$ and the immediate transitions $t_2$ and $t_3$. We assume that the 
weights of the immediate transitions have already been transformed to
probabilities. The 
resulting Labelled Transition System -- including both exponential and probabilistic transitions --
of the 
reachable markings is shown in Fig.~\ref{fig:markings}: the solid arc defines
an exponential transition with rate $\lambda$,
the dashed arcs denote the immediate transitions driven by probabilities $(1-a)$ and $(a)$.
After elimination of marking $(0,1,0,0)$ we obtain the transition system in Fig.~\ref{fig:markings_eliminated}.
In GSPN terminology, the
state corresponding to marking $(0,1,0,0)$ is called ``vanishing'',
whereas in this paper, where we develop a more detailed
classification of vanishing states, it will be denoted \emph{trivially vanishing}.

\begin{figure}
  \centering
  \subfloat[GSPN example]{\label{subfig:GSPN_example}\includegraphics[width=5cm]{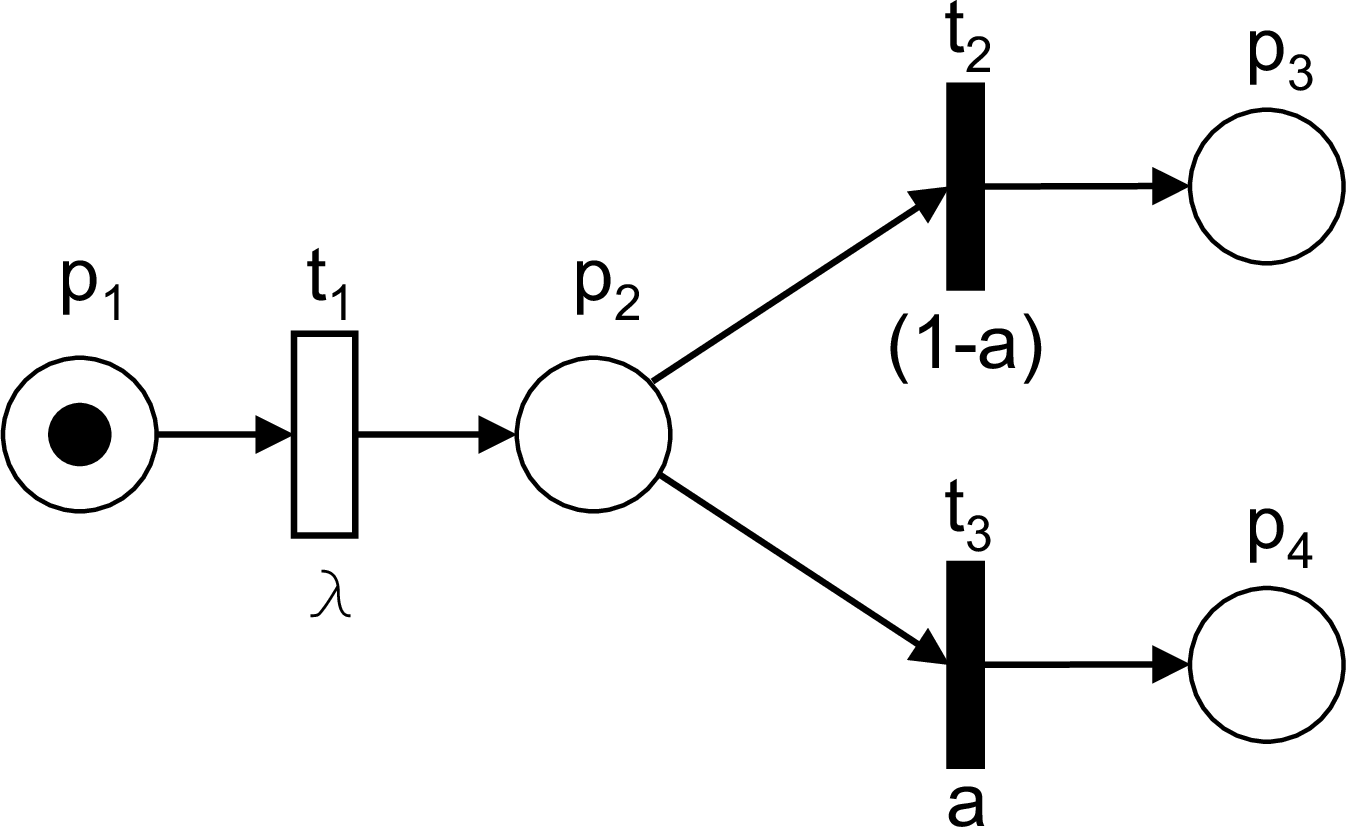}} \\
  \subfloat[Reachable markings]{\label{fig:markings}\includegraphics[height=1.5cm]{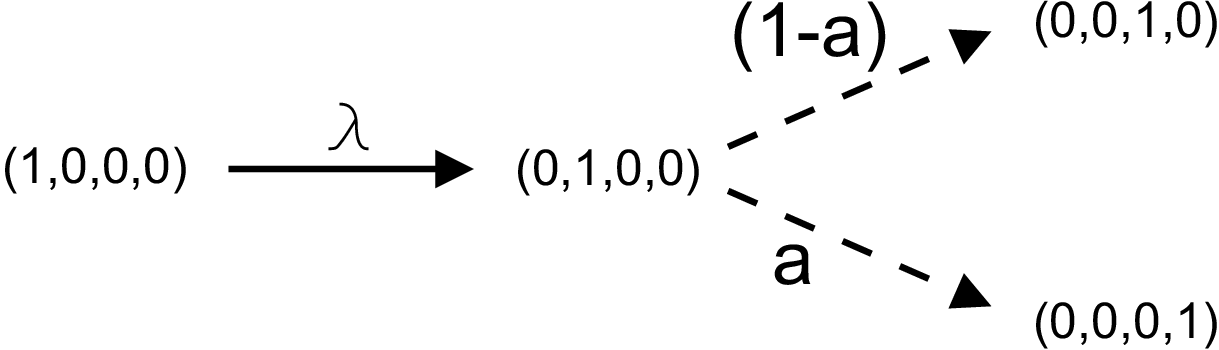}} \qquad
  \subfloat[After elimination]{\label{fig:markings_eliminated}\includegraphics[height=1.5cm]{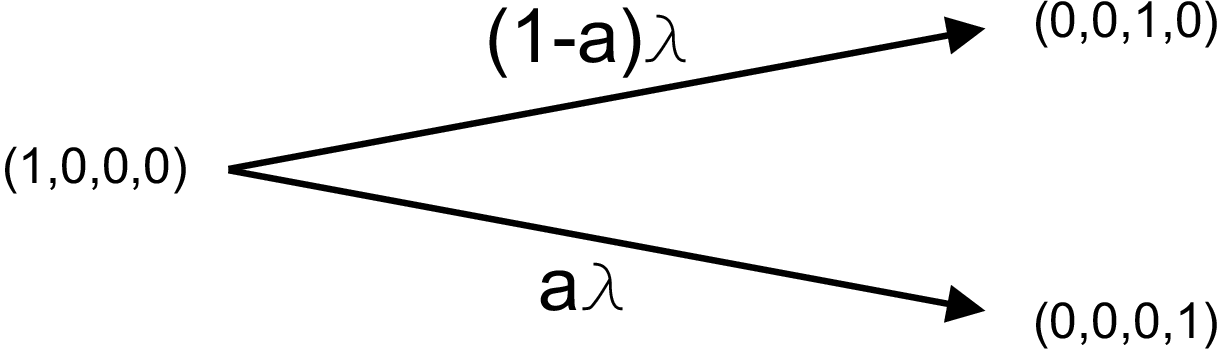}}
  \caption{Example of a GSPN
}
  \label{fig:GSPN_example}
\end{figure}

With this intuition, we are able to define vanishing states in the
nondeterministic context of MA.
We provide a topological characterisation of a special kind of states that is equivalent to a ``real''
 distribution, i.e.~a distribution consisting 
of at least two different classes with respect to some equivalence relation.
It will turn out that this characterisation, that we call \emph{non-na\"ively vanishing (nn-vanishing)}, is sufficient for calculating state-minimal 
normal forms of Markov Automata. With the aid of this characterisation we are able to give a decision algorithm 
for weak MA bisimulation.

In contrast to distribution-based weak bisimulation,
decision algorithms for  na\"ive weak bisimulation on MA have been known for some time. Since na\"ive 
weak bisimulation on MA \cite{lics:10} corresponds to weak probabilistic bisimulation on Probabilistic Automata (PA), 
na\"ive weak MA bisimulation is known to be decidable since 2002 \cite{segala:02}. There, an exponential time algorithm was presented.
In 2012 a polynomial time algorithm has been presented for deciding na\"ive weak MA bisimulation \cite{turrini:12}. 

Our algorithm is built upon the algorithm in \cite{segala:02}
which is a partition refinement algorithm.
The main difference is that for every partition
of the state space we first
identify nn-vanishing classes of states.
Before we split the current partition, we ``virtually''
 eliminate all states 
that belong to nn-vanishing classes,
i.e.~we only consider restricted 
probability distributions where the 
probability of every nn-vanishing state is equal to zero.
On this ``reduced''
transition system, we run the algorithm of \cite{segala:02} 
(on \emph{all} states, considering also nn-vanishing states which are 
sufficiently identified by the not nn-vanishing states they can reach).
This basic scheme of the algorithm is depicted in Fig.~\ref{fig:algo_schematic}.

Our algorithm has exponential time complexity (the result of \cite{turrini:12}
does not seem to be applicable to the weak case, as the na\"ive weak bisimulation problems after speculative eliminations 
have an exponential number of transitions in contrast to the original weak bisimulation problem).

By its generality, using only minor changes it can also be applied to the case of the MA bisimulation recently defined in \cite{deng-hennessy:2011}.

More or less at the same time to our approach \cite{schuster:13}, a different approach for a decision algorithm has been given in \cite{qest:13}.
A comparison of the two approaches is given in Sec.~\ref{sec:saarbruecken_small}.

\begin{figure}
\begin{center}
\includegraphics[width=10cm]{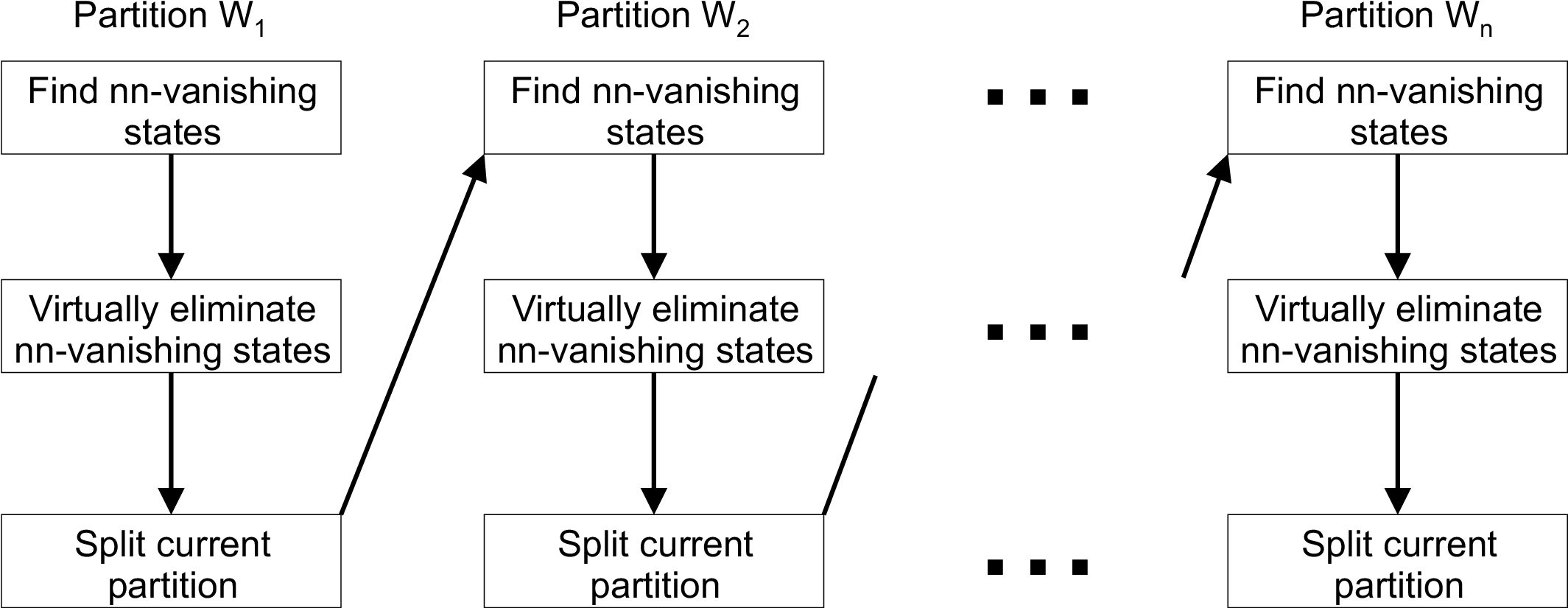}
\caption{Proposed algorithm
}
\label{fig:algo_schematic}
\end{center}
\end{figure}

This paper is a major rework of our report \cite{schuster:13}.
While the algorithm is completely the same, the correctness proofs of the previous paper
relied heavily on results of \cite{lics:10,avacsreport}.
The present paper provides a new line of argumentation and new proofs,
independent of Thm.~2 of \cite{lics:10,avacsreport}
.

The paper is organised as follows: In Sec.~\ref{sec:prelim}
we present the necessary preliminaries
and recall a mapping from MA to PA from \cite{lics:10}.
Sec.~\ref{sec:bisim} recapitulates some facts on weak and na\"ive weak bisimulation for MA. In Sec.~\ref{sec:elimrelate} we define different notions of vanishing 
states and use them to relate weak bisimulation and na\"ive weak bisimulation.
Sec.~\ref{subsec:canon_van_repres} discusses properties of 
vanishing states and provides the main theorems.
Sec.~\ref{sec:bisimalgo} describes our
decision algorithm for weak MA bisimulation that heavily relies on Sec.~\ref{sec:elimrelate}
and Sec.~\ref{subsec:canon_van_repres}.
In Sec.~\ref{sec:deng-hennessy} we briefly discuss the applicability of our concepts and 
our decision algorithm to the weak bisimulation published 
in \cite{deng-hennessy:2011}.
In Sec.~\ref{sec:saarbruecken_small} we compare our concepts to a recently published alternative approach \cite{qest:13} for deciding weak MA bisimulation.
Finally, Sec.~\ref{sec:concl} concludes the paper.

\section{Preliminaries}
\label{sec:prelim}
This section introduces some common notations on distributions, defines Markov Automata following
\cite{lics:10,avacsreport} and recalls the mapping from Markov Automata to probabilistic automata
which is used in Sec.~\ref{sec:bisim} to define weak bisimulation for Markov Automata.

\subsection{Probability (Sub-)Distributions}
First we define the notion of discrete subdistribution and related terms and notations:
A mapping $\mu: S \rightarrow [0,1]$ is called (discrete) subdistribution, if
$\sum_{s\in S}\mu(s)\leq 1$. As usual we write $\mu(S')$ for $\sum_{s\in S'}\mu(s)$.
The \emph{support} of $\mu$ is defined as 
$Supp(\mu):=\setcond{s\in S}{\mu(s)>0}$.
The empty subdistribution $\mu_\emptyset$ is defined by $Supp(\mu_\emptyset)=\emptyset$.
The \emph{size} of $\mu$ is defined as $|\mu|:=\mu(S)$. A subdistribution $\mu$ is
called \emph{distribution} if $|\mu|=1$. The sets $\Dist{S}$ and $\SubDist{S}$ denote
distributions and subdistributions defined over the set $S$.
Let $\De_s\in \Dist{S}$ denote the \emph{Dirac} distribution on $s$, i.e.~$\De_{s}(s)=1$.
For two subdistributions $\mu$, $\mu'$ the sum $\mu'':=\mu\oplus \mu'$ is defined
as $\mu''(s):=\mu(s)+\mu'(s)$ (as long as $|\mu''|\leq 1$).
As long as $c\cdot |\mu|\leq 1$, we denote by $c\mu$ the subdistribution 
defined by $(c\mu)(s):=c\cdot \mu(s)$. For a subdistribution $\mu$ and a state $s\in Supp(\mu)$
we define $\mu-s$ by 
$$(\mu-s)(t)=\begin{cases}
               \mu(t) & \text{ for }t\neq s \\
               0      & \text{ for }t = s
             \end{cases}$$
%
Occasionally, we will also need the lifting of relations to distributions:
\begin{definition}[Lifting of equivalence relations to distributions]
\label{def:lifting}
An equivalence relation $R\subseteq S\times S$ is lifted to $Dist(S)\times Dist(S)$ in the following way:
For $\mu, \gamma\in Dist(S)$ we write $\mu \equiv_R \gamma$ (or simply, by abuse of notation, $\mu \mathbin{R} \gamma$) if and only
if for each equivalence class $C\in \nicefrac{S}{R}: \mu(C)=\gamma(C)$.
\end{definition}
%
\subsection{Markov and Probabilistic Automata}
The definition of Markov Automata we use is the one from \cite{lics:10,avacsreport}. 
\begin{definition}[Markov Automata \cite{lics:10}]
  \label{def:ma_lics}
  A Markov automaton MA is a tuple $(S, Act, \mapr{ }, \mati{ }, s_0)$, where 
  \begin{itemize}
    \item $S$ is a nonempty finite set of states,
    \item $Act$ is a set of actions containing the internal action $\tau$,
    \item $\mapr{}  \subseteq S\times Act \times \Dist{S}$ a set of action-labelled probabilistic transitions,
    \item $\mati{} \subseteq S \times \mathbb{R}_{\geq 0}\times S$ a set of Markovian timed transitions and 
    \item $s_0\in S$ the initial state.
  \end{itemize}
A state in a MA is called \emph{stable}
if it has no emanating $\tau$ transitions, otherwise it is called unstable. A stable state $s$ will be denoted by
$s\hspace{-0.1cm}\downarrow$.
\end{definition}

In order to make our decision algorithm feasible we assume in the following that, in contrast to the original definition from \cite{lics:10,avacsreport}, \emph{all} 
sets in Definition~\ref{def:ma_lics} are finite. 
This means that there are finitely many states, finitely many actions and finitely many transitions.

For simplicity we define probabilistic automata (PA) in terms of MA.
\begin{definition}
A probabilistic automaton (PA) is a MA $P=(S, Act, \mathord{\rightarrow}, \emptyset, s_0)$. We also write $P=(S, Act, \rightarrow, s_0)$
if the context is clear.
\end{definition}
This definition corresponds to a \emph{simple} probabilistic automaton in the sense of Segala \cite{segala:95}.
%

For the mapping from MA to PA introduced in \cite{lics:10} we need to define the probability distribution on successor states.
In contrast to \cite{lics:10,avacsreport,hatefi:12}, our definition of successor distribution also takes care of the case $rate(s) = 0$. 
\begin{definition}[modified\footnote{The original definition from \cite{lics:10,avacsreport,hatefi:12} is problematic, 
as for $rate(s)=0$ the fraction $\frac{0}{0}$ is not defined (this case is treated separately in our definition), 
and for infinite sets $\mati{}$ the exit rate may not converge (this case is not problematic for us, as we deal with \emph{finite} sets).
Both issues have no impact on the decision algorithm presented here, 
but the first issue has an impact on the compositionality of MA bisimulation in general. For a detailed explanation of why compositionality is lost with the original
definitions of \cite{lics:10,avacsreport,hatefi:12}
we refer to Appendix A of \cite{schuster:13}.} 
version of Definition 3 in \cite{lics:10}]
\label{def:chi0}
Let $M=(S, Act, \mathord{\mapr{ }}, \mathord{\mati{ }}, s_0)$ be a MA. Define 
$$rate(s,s'):=\sum_{(s,\lambda, s')\in \mati{}}\lambda$$
and $rate(s):=\sum_{s'\in S}rate(s,s')$ which is called the \emph{exit rate}
of state $s$.
The probability
distributions $P_s$ are defined in the following way:
$$
P_s:=\begin{cases}s'\mapsto \frac{rate(s,s')}{rate(s)} & \text{ for } rate(s) \neq 0 \\
                       \De_s                         & \text{ otherwise}
        \end{cases}
$$
\end{definition}
%
%

\subsection{A mapping from MA to PA}

The remarkable idea of \cite{lics:10} is to define bisimulations on MA using a mapping from MA to PA. The basic ingredient is a set of special actions,
denoted by $\chi(.)$, that cover timed behaviour. 
In the setting of \cite{lics:10,avacsreport} countable action sets are mapped to uncountable action sets by definition, as for every
real number a new action name is introduced.
In order to keep the action set finite to retain algorithmic tractability, we redefine $Act^{\chi}$ in the context of a fixed MA:
\begin{definition}
Let $M=(S, Act, \mapr{ }, \mati{ }, s_0)$ be a MA.
Assume $\forall r\in \mathbb{R}_{\geq 0} \text{ }\chi(r)\notin Act$ and define
$\mathcal{RT}:=\{rate(s)|s\in S\}$ (which is finite). Then we define
$Act^{\chi(\mathcal{RT})}=\{\chi(r)|r\in \mathcal{RT}\}$ and
$Act^{\chi}:=Act\cup Act^{\chi(\mathcal{RT})}$.
\end{definition}

There is a mapping from MA to PA 
(adapted from \cite{lics:10}) where we use Definition \ref{def:chi0}:
\begin{definition}
Let $M=(S, Act, \mapr{ }, \mati{ }, s_0)$ be a MA. Define the transitions $\rightarrow$ as follows:
For $s\in S$ define
\[
s \stackrel{\alpha}{\rightarrow} \mu\text{ if }\begin{cases} \alpha \in Act \text{ and } s\mapr{\alpha} \mu  \\
                                      s\hspace{-0.1cm}\downarrow, \alpha=\chi(rate(s)) \in Act^{\chi(\mathcal{RT})} \text{ and } \mu=P_s
                        \end{cases}   
\]
Then the mapping $\mapping : MA \rightarrow PA$ is defined by $M\mapsto (S, Act^\chi, \rightarrow, \emptyset, s_0)$.
\end{definition}

Note that every timed transition is part of a special $\chi(\cdot)$ action. So the set of actions is increased by the mapping $\mapping (\cdot)$, but no
timed transition remains in the image. For more details on the procedure we refer to \cite{lics:10}.
%

\begin{figure}
  \centering
  \subfloat[\mbox{$M_1=\mapping(M_1)$}]{\hspace{0.7cm}\label{fig:labelled}\includegraphics[height=1.4cm]{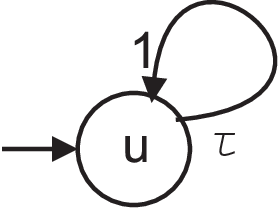}\hspace{0.4cm}} \qquad 
  \subfloat[$M_2$]{\label{fig:faith}\includegraphics[height=1.4cm]{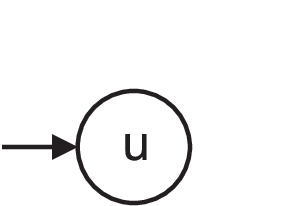}} 
  \subfloat[$\mapping$($M_2$)]{\label{fig:nonfaith}\includegraphics[height=1.4cm]{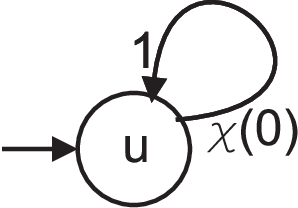}}
  \caption{MA to PA transformations}
  \label{fig:mapa}
\end{figure}

\begin{example}
Two 
special cases are given in Fig.~\ref{fig:mapa}. $M_1$ (Fig.~\ref{fig:labelled}) has a $\tau$ loop, so no timed transition (i.e.~$\chi(\cdot)$) 
exists after the transformation, i.e.~
$\mapping (M_1)=M_1$. In the example $M_2$ (Fig.~\ref{fig:faith}) $u$ is a stable state and therefore the transformation 
leads to a $\chi$ transition with exit rate 0 (Fig.~\ref{fig:nonfaith}).
\end{example}


\begin{figure}
  \centering
  \subfloat[$M_3$]{\label{fig:trap2}\includegraphics[height=1.4cm]{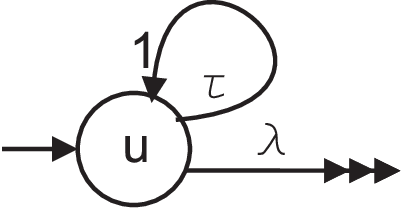}} \qquad
  \subfloat[PA which is not in $\mapping (MA)$]{\label{fig:manotinpa}\includegraphics[width=7.5cm]{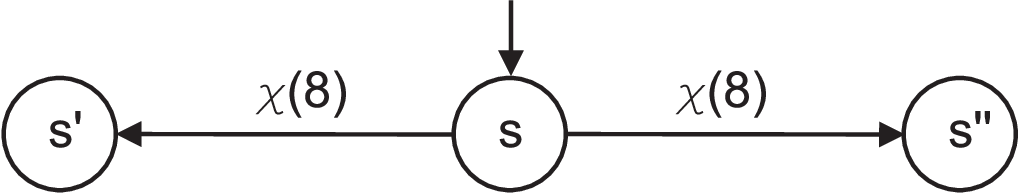}}
  \caption{MA to PA transformations - continued}
  \label{fig:mapa2}
\end{figure}

\begin{lemma}
The mapping $\mapping : MA \rightarrow PA$ is neither injective nor surjective.
\end{lemma}

\begin{proof}
It is not surjective, as a Markovian race condition is always converted to
a deterministic $\chi$ transition. That means for example the PA in Fig.~\ref{fig:manotinpa} is not in $\mapping (MA)$.
It is also not injective, as $M_3$ in Fig.~\ref{fig:trap2} is not equal to $M_1$ in Fig.~\ref{fig:labelled}, but $\mapping (M_1)=\mapping (M_3)=M_1$.
\end{proof}




\subsection{Weak transitions}
\label{sec:weaktrans}
In the following we use the definitions and terminology of \cite{Segala:07}, but we leave out the definitions for labelled transition systems.
Given a transition $tr = (s, a, \mu)$, we denote $s$ by $source(tr)$ and $\mu$ 
by $\mu_{tr}$.
Consider a PA $\aPA$ (with transition relation $T$).
An execution fragment of $P$ is a finite or infinite
sequence $\alpha = q_0 a_1 q_1 a_2 q_2\cdots$ of alternating states and actions, starting with a state
and, if the sequence is finite, ending in a state, where each $(q_i,a_{i+1},\mu_{i+1})\in T$
and $\mu_{i+1}(q_{i+1})>0$.
State $q_0$, the first state of $\alpha$, is denoted by $fstate(\alpha)$. If $\alpha$ is a finite
sequence, then the last state of $\alpha$ is denoted by $lstate(\alpha)$.
An \emph{execution} of $P$ is an execution fragment of $P$ where $q_0=s_0$.
We let $frags(P)$ denote the set of execution fragments of $P$ and $frags^\ast(P)$
the set of finite execution fragments of $P$. Similarly, we let $execs(P)$
denote the set of executions of $P$ and $execs^\ast(P)$ the set of finite executions.
Execution fragment $\alpha$ is a \emph{prefix} of execution fragment $\alpha'$, denoted
$\alpha \leq \alpha'$, if sequence $\alpha$ is a prefix of sequence $\alpha'$.

The \emph{trace} of an execution fragment $\alpha$,
written $trace(\alpha)$, is the
sequence of actions obtained by restricting $\alpha$ to the set of external actions, i.e.~$Act\setminus \{\tau\}$.
For a set $E$ of executions of a PA $P$, $traces(E)$
is the set of traces of the executions in $E$. 
We say that $\beta$ is a trace of a PA $P$ if there is an execution $\alpha$ of $P$ with
$trace(\alpha)=\beta$. 
Let $traces(P)$ denote the set of traces of $P$. 

A \emph{scheduler} for a PA $P$ is a function $\sigma: frags^\ast(P)\rightarrow \SubDist{T}$
such that $tr\in supp(\sigma(\alpha))$ implies that $source(tr)=lstate(\alpha)$.
This means that the image $\sigma(\alpha)$ is a \emph{discrete} subdistribution over transitions.
The defect of the subdistribution, i.e.~$1-|\sigma(\alpha)|$ is used for stopping in the current state. 
A scheduler $\sigma$ is said to be \emph{deterministic} if for each finite execution fragment
$\alpha$ either $\sigma(\alpha)(T)=0$ or $\sigma(\alpha)=\Delta_{tr}$ (Dirac measure for $tr$) for some $tr\in T$.
In other words,
a deterministic scheduler is the entity that resolves nondeterminism in a probabilistic automaton by choosing randomly either
to stop or to perform one of the transitions that are enabled from the current state.
A scheduler is called \emph{memoryless} if it depends only on the last state of its argument, that is, for each pair $\alpha_1$,
$\alpha_2$ of finite execution fragments, if $lstate(\alpha_1)=lstate(\alpha_2)$, then $\sigma(\alpha_1)=\sigma(\alpha_2)$.
A scheduler is called \emph{determinate} if its choice depends only on the current trace and on the last state of its argument, that is,
for each pair $\alpha_1$, $\alpha_2$ of finite execution fragments, if $trace(\alpha_1)=trace(\alpha_2)$ and $lstate(\alpha_1)=lstate(\alpha_2)$, then $\sigma(\alpha_1)=\sigma(\alpha_2)$.
Following \cite{segala:02} we call a deterministic determinate scheduler a \emph{Dirac determinate} scheduler.

A scheduler $\sigma$ and a discrete initial probability measure $\mu_0\in \Dist{S}$ induce a measure $\epsilon$ on the sigma-field
generated by cones of execution fragments as follows. If $\alpha$ is a finite execution fragment, then the \emph{cone}
of $\alpha$ is defined by $C_\alpha=\{\alpha' \in frags(P)| \alpha \leq \alpha' \}$.
The measure $\epsilon$ of a cone $C_\alpha$ is defined recursively:
If $\alpha=s$ for some $s\in S$ we define $\epsilon(C_\alpha)=\mu_0(s)$.
If $\alpha$ is of the
form $\alpha' a' s'$, $\epsilon(C_\alpha)$ is defined by the equation
$$\epsilon(C_\alpha)=\epsilon(C_{\alpha'})\cdot \sum_{tr \in T(a')}\sigma(\alpha')(tr)\mu_{tr}(s'),$$
where $T(a')$ denotes the set of transitions of $T$ that are labelled by $a'$. Standard measure theoretical arguments
ensure that $\epsilon$ is well defined. We call the measure $\epsilon$ a probabilistic execution fragment of $P$, and we say
that $\epsilon$ is generated by $\sigma$ and $\mu_0$.

Consider a probabilistic execution fragment $\epsilon$ of a PA $P$, with first state $s$, i.e.~$\mu_0=\Delta_{s}$, that assigns probability
1 to the set of all finite execution fragments $\alpha$ with trace $trace(\alpha) = \beta$ for some $\beta \in (Act \setminus \{\tau\})^\ast$. Let $\mu$ be the discrete
measure defined by $\mu(s')=\epsilon(\{\alpha | lstate(\alpha)=s'\})$.
Then $s\stackrel{\beta}{\Rightarrow}_C\mu$ is a \emph{weak combined transition} of $P$. We call $\epsilon$ a \emph{representation}
of $s\stackrel{\beta}{\Rightarrow}_C\mu$. If $s\stackrel{\beta}{\Rightarrow}_C\mu$ is induced by a deterministic scheduler, we also write
$s\stackrel{\beta}{\Rightarrow}\mu$.
In case $trace(\alpha)$ is empty we write $s\stackrel{\tau}{\Rightarrow}_C\mu$.

Let $\{s\stackrel{a}{\rightarrow}\mu_i\}_{i\in I}$ be a collection of transitions of a PA $P$, and let $\{c_i\}_{i\in I}$ be a collection
of probabilities such that $\sum_{i\in I}c_i=1$. Then the triple
$(s,a,\sum_{i\in I}c_i\mu_i)$ is called a \emph{(strong) combined transition} of $P$
and we write $s\stackrel{a}{\rightarrow}_C \sum_{i\in I}c_i \mu_i$.
%
We say that there is a \emph{hyper-transition} from $\mu\stackrel{a}{\Rightarrow}_C \mu'$, 
if there exists a family of weak combined transitions
$\{s\stackrel{a}{\Rightarrow}_C \mu_s\}_{s\in Supp(\mu)}$ such that $\mu' = \sum_{s\in Supp(\mu)}\mu(s)\cdot \mu_s$.

\section{Relating na\"ive weak \& weak bisimulation}
\label{sec:bisim}
Remember that for a MA $M$ its transitions have been defined by means of 
$\mapping (M)$, so in the following it is safe to assume that all Markov Automata
are represented
by their PA images. All calculations will be made in this context.
%
%

Note that, in contrast to the transition tree notation of \cite{lics:10,avacsreport,stochworld}, we do not need the notation 
$\stackrel{\hat{\alpha}}{\Rightarrow}$ (which includes the possibility of zero steps in the case $\alpha=\tau$) as our definition 
of $\stackrel{\alpha}{\Rightarrow}$ also includes this case\footnote{Definition 10 in \cite{lics:10,avacsreport} erroneously uses $\alpha$ instead of $\hat{\alpha}$.
}.


The relation defined in the following is called ``weak probabilistic bisimulation'' \cite{segala:02} in the context of PA:

\begin{definition}[Na\"ive weak bisimulation in the spirit of \cite{segala:95b}]
\label{def:naiveweak}
An equivalence relation $\mathcal{R}$ on the set of states $S$ of a MA $M=(S, Act, \mapr{ }, \mati{}, s_0)$ is called \emph{na\"ive weak bisimulation} if and only if
$x \mathcal{R} y$ implies for all $\alpha \in Act^\chi$: $(x\stackrel{\alpha}{\rightarrow}\mu)$ implies $(y\stackrel{\alpha}{\Rightarrow}_C\mu')$ 
with $\mu \equiv_R \mu'$
(note that the transitions are regarded in $\mapping (M)$).
If $x$ and $y$ are contained in a na\"ive weak bisimulation relation, we write
$x \approx_{\text{na\"ive}} y$.
Two MA are called \emph{na\"ively weakly bisimilar} if
their initial states are related by a na\"ive weak bisimulation relation on the direct sum of their states.
\end{definition}

We would like to mention that modulo na\"ive weak bisimulation it is possible to omit $\tau$-loops in the image $\mapping (\cdot)$. The property whether
a state is stable or unstable can still be recovered by looking for the presence (or absence) of $\chi$ transitions.

The authors of \cite{lics:10,avacsreport} argued that the (state-based) notion of na\"ive weak bisimulation is too fine.
Therefore they defined the coarser notion of (distribution-based) weak bisimulation:

\begin{definition}[Weak bisimulation \cite{lics:10,avacsreport}]
\label{def:weakbisim}
A relation $\mathcal{R}$ on sub-distributions over the set of states $S$ of a MA $M=(S, Act, \mapr{ }, \mati{}, s_0)$ is called \emph{weak bisimulation} if 
for all $(\mu_1,\mu_2)\in \mathcal{R}$ it holds that (transitions regarded in $\mapping (M)$) 
\begin{enumerate}[A.)]
 \item $|\mu_1|=|\mu_2|$
 \item 
       $\forall t \in Supp(\mu_1), \forall \alpha \in Act^\chi: \exists \mu_2^g, \mu_2^s: \mu_2 \Rightarrow_C \mu_2^g \oplus \mu_2^s$ such that 
      \begin{enumerate}[(i)]
        \item $(\mu_1(t)\cdot \De_t) \mathcal{R} \mu_2^g$ and $(\mu_1-t)\mathcal{R}\mu_2^s$
        \item whenever $(t\stackrel{\alpha}{\rightarrow}\mu'_1)$, then $(\exists \mu_2': \mu_2^g \stackrel{\alpha}{\Rightarrow}_C\mu_2'$ and  $(\mu_1(t)\cdot \mu'_1) \mathcal{R}\mu_2')$
      \end{enumerate}
 \item a symmetric condition with $\mu_1$ and $\mu_2$ interchanged (roles of left-hand side and right-hand side also interchanged)
\end{enumerate}
Two distributions $\mu$, $\gamma$ are called \emph{weakly bisimilar} (with respect to some MA $M$), written $\mu \approx \gamma$, if the pair $(\mu, \gamma)$ is contained in a weak bisimulation
relation (with respect to $M$). Two states are called weakly bisimilar if their corresponding Dirac distributions are weakly bisimilar. We write $s\approx_\De t$ for $\De_s \approx \De_t$.
Two MA are called weakly bisimilar if their initial states are weakly bisimilar in the direct sum of the MA.
\end{definition}

Thm.~1 in \cite{lics:10,avacsreport} shows that $\approx$ is an equivalence relation (and therefore also $\approx_\De$).



The following statement is a corollary of Thm.~2 in \cite{lics:10,avacsreport}.
As we show in \ref{sec:saarbruecken_continuous} that Lemma 16 in \cite{avacsreport} (and therefore also Thm.~2 in \cite{lics:10,avacsreport})
must be considered unproven,
we give an independent proof of the following statement.
\begin{lemma}[Corollary of Thm.~2 in \cite{lics:10,avacsreport}]
\label{lemma:naive-is-weak}
If two MA are na\"ively weakly bisimilar, they are also weakly bisimilar.
\end{lemma}

\begin{proof}
We provide a direct proof of this statement.
Let $P=(S, Act^\chi, \rightarrow, \emptyset, s_0)$, $P'=(S', Act^\chi, \rightarrow', \emptyset, s'_0)$.
We directly construct a weak bisimulation relation $R'$ out of a given na\"ive weak bisimulation relation $R \subseteq S \times S'$.
$R'=\setcond{(\mu_1, \mu_2)}{\mu_1\in \SubDist{S}, \mu_2\in \SubDist{S'}, \mu_1\equiv_R \mu_2\ }$.
We show that this is indeed a weak bisimulation.
To simplify matters, we assume that we work on the quotient
with respect to $\approx_{\text{na\"ive}}$
of the direct sum of $S$ and $S'$.
Now we may verify the condition for weak bisimulation. Assume that $(\mu_1, \mu_2) \in R'$. By $\mu_1 \equiv_R \mu_2$ it is clear that $|\mu_1|=|\mu_2|$.
Now choose an arbitrary $s\in Supp(\mu_1)$ and let $\mu_1(s)=c$. Again by $\mu_1 \equiv_R \mu_2$ we find also $s\in Supp(\mu_2)$
and we may choose $\mu_2^g=c\Delta_{s}$
(the fact that we may consider the same state $s$ in both $Supp(\mu_1)$ and $Supp(\mu_2$)
 holds only for quotients -- in the general
case we must find bisimilar states with the same probability mass $c$). Of course also the remaining parts $\mu_1-s$ and $\mu_2^s$ 
(where $\mu_2=\mu_2^g \oplus \mu_2^s$) must be in relation, 
as $\mu_1 \equiv_R \mu_2$.
Finally it is clear that whenever $s\stackrel{a}{\rightarrow}{\mu'_1}$, then also $\mu_2^g\stackrel{a}{\Rightarrow}_C \mu''$ with $(c\cdot \mu'_1, \mu'')\in R'$ by na\"ive weak bisimilarity and the construction of $R'$.
The same holds with the roles of $\mu_1$ and $\mu_2$ interchanged, so the claim is shown.
\end{proof}


A precalculation of the elimination procedure presented later in this paper is the ``rescaling'' procedure. 
\begin{remark}
\label{rem:rescale}
Combined transitions $\Rightarrow_C$ can be used to rescale loops. A basic example is given in Fig.~\ref{fig:rescale_loop}. 
A Dirac determinate scheduler choosing the transition $s\stackrel{\tau}{\rightarrow}\frac{1}{3}\Delta_{x}\oplus \frac{1}{3}\Delta_{s}\oplus \frac{1}{3}\Delta_{y}$ with probability one and 
stopping in states $x$ and $y$
leads to the 
weak transition
$s\stackrel{\tau}{\Rightarrow}\frac{1}{2}\De_x \oplus \frac{1}{2}\De_y$. 
To mimic the $\tau$ transition of $s$, $t$ has to perform a transition combined of $\frac{2}{3}$ times $t\rightarrow \frac{1}{2}\De_x \oplus \frac{1}{2}\De_y$
and $\frac{1}{3}$ times $t \Rightarrow \De_t$. Without using combined transitions $t$ could not mimick this transition of $s$.
\end{remark}

\begin{figure}
  \centering
  \includegraphics[width=8cm]{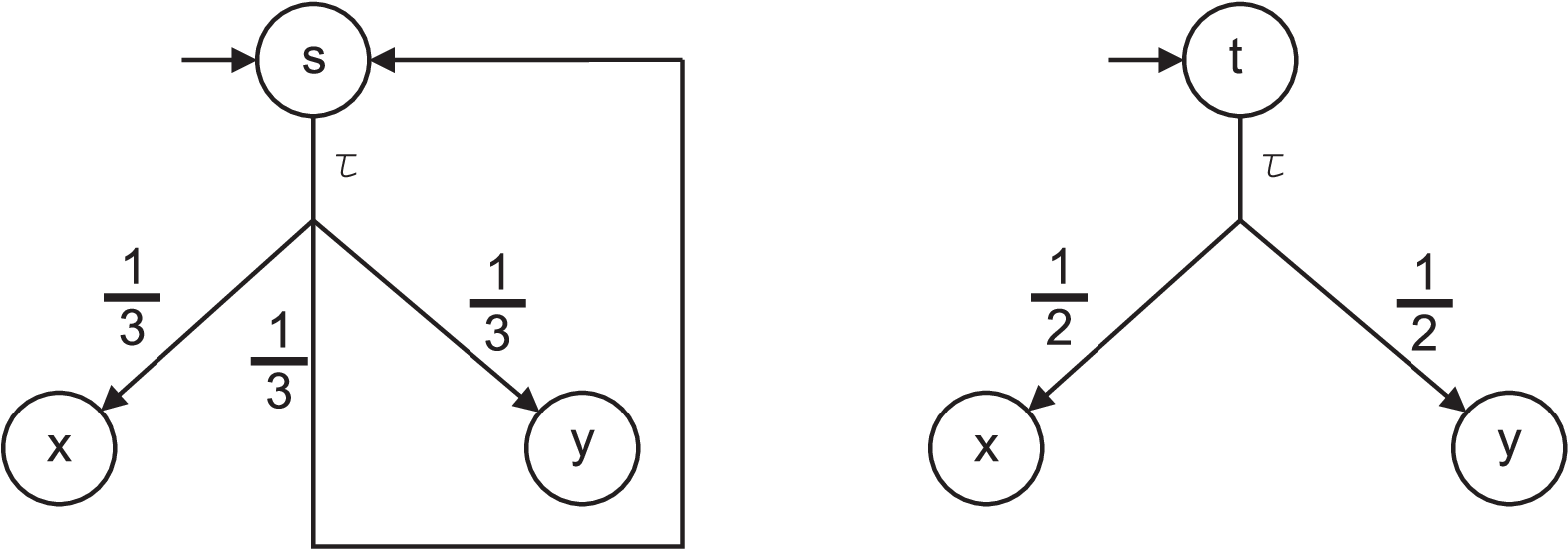}\\
  \caption{Resolving a $\tau$ loop}
  \label{fig:rescale_loop}
\end{figure}

For many of the proofs in this paper we are only interested in properties ``up to an equivalence relation $R$'', which motivates the following definition.
In the quotient automaton,  no two distinct states $s$, $t$ exist with $sR\,t$.

\begin{definition}[Quotient automaton]
Let $\aPA$ be a PA 
and $R$ an equivalence relation over 
$\stateSet$.
The equivalence class of a state $s$ is denoted by $[s]_{R}$
(when the context is clear we write $[s]$ instead of $[s]_R$).
We write $\nicefrac{P}{R}$ to denote the 
quotient automaton of $P$ with respect to $R$, that is
$$\nicefrac{P}{R}=(\nicefrac{\stateSet}{R},Act,\nicefrac{\transitionRelation}{R},[\startState]_R)$$
with $\nicefrac{\transitionRelation}{R}\subseteq \nicefrac{\stateSet}{R}\times Act \times \Dist{\nicefrac{\stateSet}{R}}$ such that
$([s]_R,a,\mu)\in \nicefrac{\transitionRelation}{R}$ if and only if there exists a state
$s' \in [s]_R$ such that $(s',a,\mu')\in \transitionRelation$ and $\forall [t]_R\in \nicefrac{\stateSet}{R}: \mu([t]_R)=\sum_{t'\in [t]_R}\mu'(t')$.
We call an automaton a \emph{quotient with respect to $R$} or, if $R$ is clear, just a \emph{quotient}, if it holds that 
$P$ and $\nicefrac{P}{R}$ coincide (up to a renaming of the set of states).
\end{definition}

\begin{definition}[Reachable states]
Let $\aPA$ be a PA, $\stateSet' \subseteq \stateSet$ its set of reachable states, i.e.~those states that can be reached with
non-zero probability by a scheduler starting from $\startState$.
Let $\transitionRelation':=\transitionRelation|_{\stateSet'\times Act \times \Dist{\stateSet'}}$
 be the restriction of the transition relation to $\stateSet'$.
We define $r(P):=(\stateSet',Act, \transitionRelation', \startState)$ and call it the \emph{reachable fragment} of $P$.
\end{definition}

As bisimulation only focuses on the reachable fragment of the state space, we assume from now on that by \emph{quotient} we always mean the reachable
part of the quotient, i.e.~$r(\nicefrac{P}{R})$.

\section{Vanishing states and vanishing representations}
\label{sec:elimrelate}

We now introduce a notion of vanishing states in the context of MA.


\begin{definition}
	Given a PA $P = (S, Act, T, \emptyset, s_0)$, and $t \notin S$,
	we define the following renamings:
	\begin{itemize}
		\item 
			for each $v \in \stateSet$, 
			\[
				\replace{v}{\replacement{s}{t}} = 
				\begin{cases}
					t & \text{if $v = s$,} \\
					v & \text{otherwise.}
				\end{cases}
			\]
			The set of all renamed states is denoted by $\replace{\stateSet}{\replacement{s}{t}}$;
		\item 
			for each $\sd \in \Dist{\stateSet}$ and $v \in \replace{\stateSet}{\replacement{s}{t}}$, 
			\[
				\probeval{\replace{\sd}{\replacement{s}{t}}}{v} = 
				\begin{cases}
					\probeval{\sd}{s} & \text{if $v = t$,} \\
					\probeval{\sd}{v} & \text{otherwise;}
				\end{cases}
			\]
		\item 
			for each $(v,a,\sd) \in T$, $\replace{(v,a,\sd)}{\replacement{s}{t}} = (\replace{v}{\replacement{s}{t}},a,\replace{\sd}{\replacement{s}{t}})$.
			The set of all renamed transitions is denoted by $\replace{T}{\replacement{s}{t}}$.
	        \item   $(\replace{\stateSet}{\replacement{s}{t}}, \replace{\startState}{\replacement{s}{t}}, \actionSet, \replace{\transitionRelation}{\replacement{s}{t}})$ is denoted by $\replace{\aut}{\replacement{s}{t}}$.
	\end{itemize}
\end{definition}

\begin{definition}[Emanating Internal Weak Combined Transitions]
\label{def:emanatingWeakTransitions}
	Given a PA $\aut = (S, Act, T, \emptyset, s_0)$ and a state $s \in S$, 
	we denote by $\transitionsCombinedFromState{s}$ the set $\setcond{(s,\hidden,\sd)}{\weakCombinedTransition{s}{\hidden}{\sd}}$ of 
	\emph{internal weak combined transitions emanating from $s$}.
	Further, we denote by $\mathfrak{S}(s)$ the set $\setcond{(s,a,\sd)}{\strongTransition{s}{a}{\sd}, a\in Act}$ of 
	\emph{strong transitions emanating from $s$}.
\end{definition}


\begin{definition}[local change of transitions]
\label{def:localChangeOfTransitions}
\label{def:transet}
Let $P=(S, Act, T, \emptyset, s_0)$ be a PA and $s \in \stateSet$. 

For any $\mathfrak{T} \subseteq \transitionsCombinedFromState{s}$ we define the PA 
\[
	\transitionReplacement{P}{\mathfrak{T}} = (S, Act, (T \setminus \mathfrak{S}(s) \cup \mathfrak{T}),\emptyset, s_0)\text{.}
\]
If $\mathfrak{T} = \setnocond{(s,\hidden,\sd)}$ we also write $\locallyChangedAut{s}{\sd}$ -- or simply $P'$ if the context is clear -- 
instead of $\transitionReplacement{P}{\setnocond{(s,\hidden,\sd)}}$.
\end{definition}

\begin{definition}[vanishing states]
\label{def:vanishing}
Let $P = (S, Act, T, \emptyset, s_0)$ be a PA. Let $s\in S$ be unstable and $\transitionsCombinedFromState{s}$ as in Def.~\ref{def:emanatingWeakTransitions}.
State $s$ is called 
\begin{description}
  \item[\emph{trivially vanishing}] 
	if $\mathfrak{S}(s) = \setnocond{(s,\hidden, \sd)}$ for some $\sd \in \Dist{\stateSet}$. 
  \item[\emph{vanishing}] 
	if there exists $(s,\hidden,\sd) \in \transitionsCombinedFromState{s}$ such that $s \approx_\Delta t$ when comparing $\replace{P}{\replacement{s}{t}}$
	and $\locallyChangedAut{s}{\sd}$ 
	for $t \notin \stateSet$. 
	In this case $P_{(s,\sd)}$ -- or $(s,\sd)$, for short -- is called a \emph{vanishing representation} of $s$.
  \item[\emph{non-na\"ively vanishing}] 
	or \emph{nn-vanishing}, for short, if it is vanishing and there is a vanishing representation $\locallyChangedAut{s}{\sd}$ such that there exists $t \in \Supp{\sd}$ such that $s \not\weakBisimD t$.
  \item[\emph{na\"ively vanishing}]
        if it is vanishing, but not nn-vanishing, i.e.~for all vanishing representations $\locallyChangedAut{s}{\sd}$ it holds that for all $t \in \Supp{\sd}:s \weakBisimD t$.
\end{description}
A state that is not vanishing is called \emph{tangible}.
A state that is not nn-vanishing is called nn-tangible\footnote{Note that in \cite{schuster:13} 
we used the term ``tangible'' for describing ``nn-tangible''. For a more readable notation we distinguish now between tangible and nn-tangible.}.
\end{definition}

Trivially vanishing states correspond to vanishing markings in GSPNs \cite{ajmone:90}, provided
that they are well-defined, i.e.\ there is no non-determinism \cite{ciardo:1996}. 
Vanishing states extend this idea to the presence of non-determinism: for a vanishing state,
emanating non-deterministic transitions can be 
bisimilarly reduced to a single deterministic transition.
Non-naively vanishing states can be transformed to a distribution, such that the equivalence class with respect to $\approx_\Delta$ changes. It will turn out
that this is essentially the difference to state-based bisimulations.

According to the definition, the set of
states can be partitioned in two ways: 
vanishing vs.\ tangible, or nn-vanishing vs.\ nn-tangible.
This classification of states is illustrated in Fig.~\ref{fig:classes1}
and Fig.~\ref{fig:classes2}.
%
By definition, any vanishing representation of a na\"ively vanishing state cannot change the equivalence class with respect to $\approx_\Delta$.
For that reason, na\"ively vanishing states do not have to be 
eliminated in order to reduce the problem to na\"ive bisimulation. We will show later that the only obstacle are nn-vanishing states.
The next example shows basic representatives of the different types of vanishing states.
\begin{figure}
  \centering
  \includegraphics[width=7.5cm]{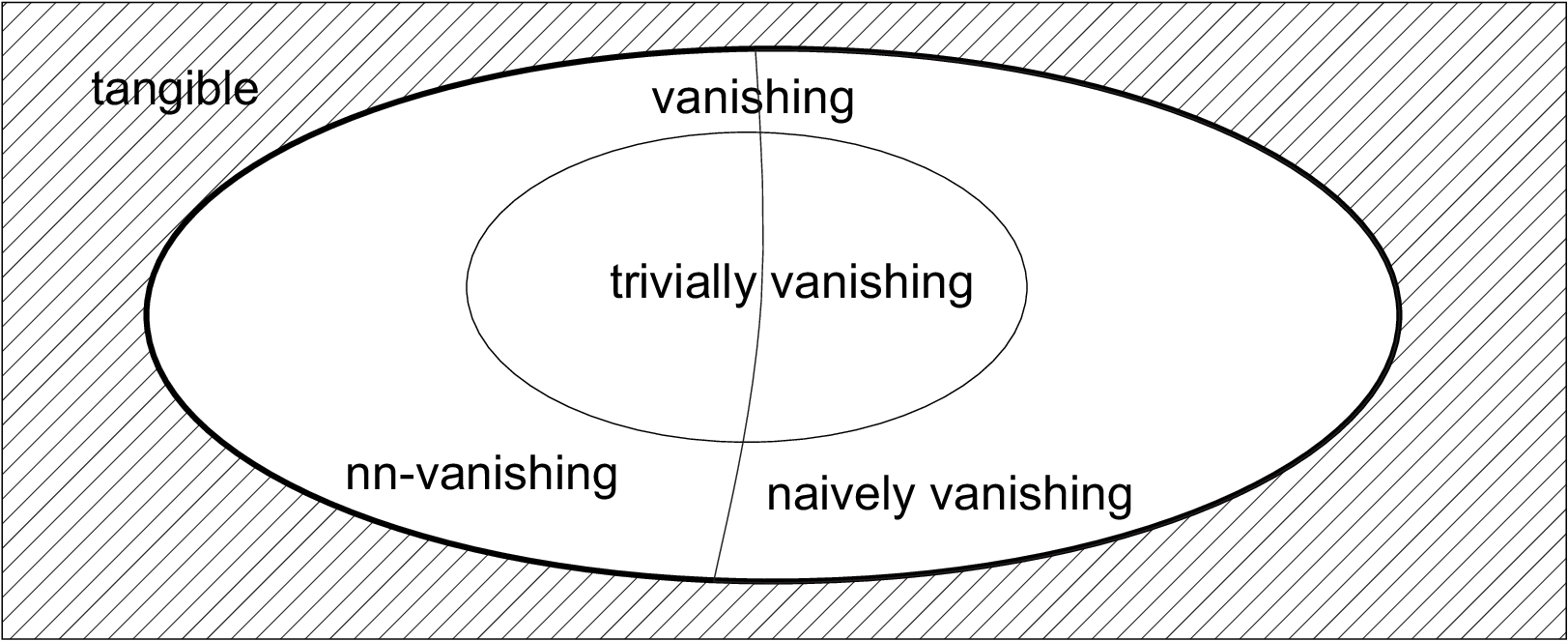}\\
  \caption{Partition of state space into vanishing and tangible states}
  \label{fig:classes1}
\end{figure}
\begin{figure}
  \centering
  \includegraphics[width=7.5cm]{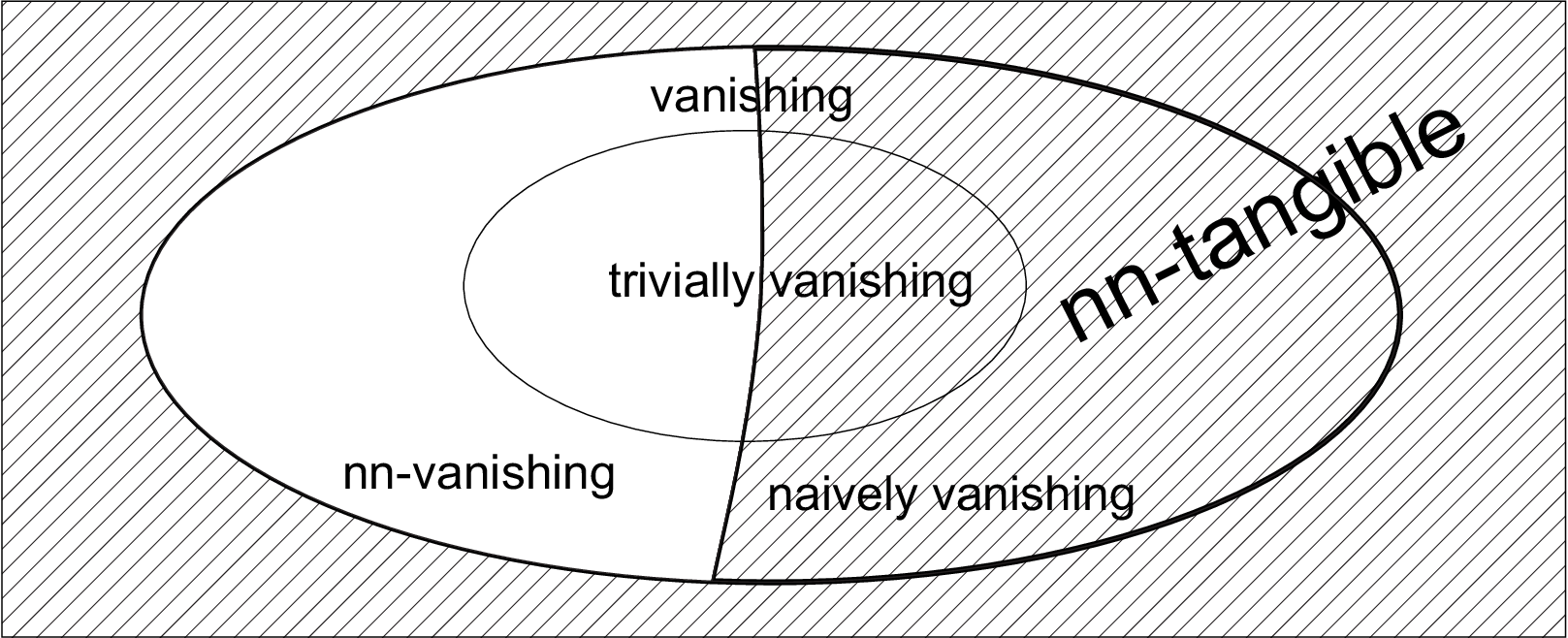}\\
  \caption{Partition of state space into nn-vanishing and nn-tangible states
}
  \label{fig:classes2}
\end{figure}

\begin{example}
Assume that $p\in (0,1)$.
State $E$ in Fig.~\ref{fig:E_triv_van} is trivially vanishing since
it only has an emanating $\tau$ transition, and it is na\"ively vanishing, as it holds that $E\approx_\Delta C\approx_\Delta D$.
A non-trivially and na\"ively vanishing state $E$ is given in Fig.~\ref{fig:E_van}.
Note that all non-$\tau$ transitions emanating from $E$ may be omitted as they
can be mimicked by appropriate weak transitions.
The automaton in Fig.~\ref{fig:E_triv_van} is the corresponding vanishing representation, so na\"ivety follows as in this case.
$E$ is na\"ively vanishing as it turns out to be in the same class as $C$ and $D$ modulo weak bisimulation.
%
For the next example, first note that in Fig.~\ref{fig:E_nn_van} 
$C$ and $D$ cannot be weakly bisimilar (because $C$ can only perform the $c$ to $A$, while $D$ can additionally perform the $d$ to $B$).
As $E$ is trivially vanishing we notice that it is also nn-vanishing, because $E$ moves to the distribution $p\De_C\oplus (1-p)\De_D$, where $C\not \approx D$.
Moreover, since in Fig.~\ref{fig:E_nn_van} $D$ is not vanishing (and $E$ is
nn-vanishing), we have that $E\not\approx D$.
%
In the last example in Fig.~\ref{fig:E_nn_van_nt} we see that $E$ is not trivially vanishing as there is more than one emanating transition.
Still it is easy to verify that the automaton in Fig.~\ref{fig:E_nn_van} is a vanishing representation, as the Dirac determinate scheduler choosing the transition
$E \rightarrow p\Delta_C\oplus (1-p)\Delta_D$ with probability 1, $D \rightarrow \Delta_E$ with probability 1, and stopping in all other states realises the transition
$E\stackrel{\tau}{\Rightarrow}\Delta_C$.
\end{example}

\begin{figure}
  \centering
  \subfloat[E trivially \& na\"ively vanishing]{\label{fig:E_triv_van}\includegraphics[width=5cm]{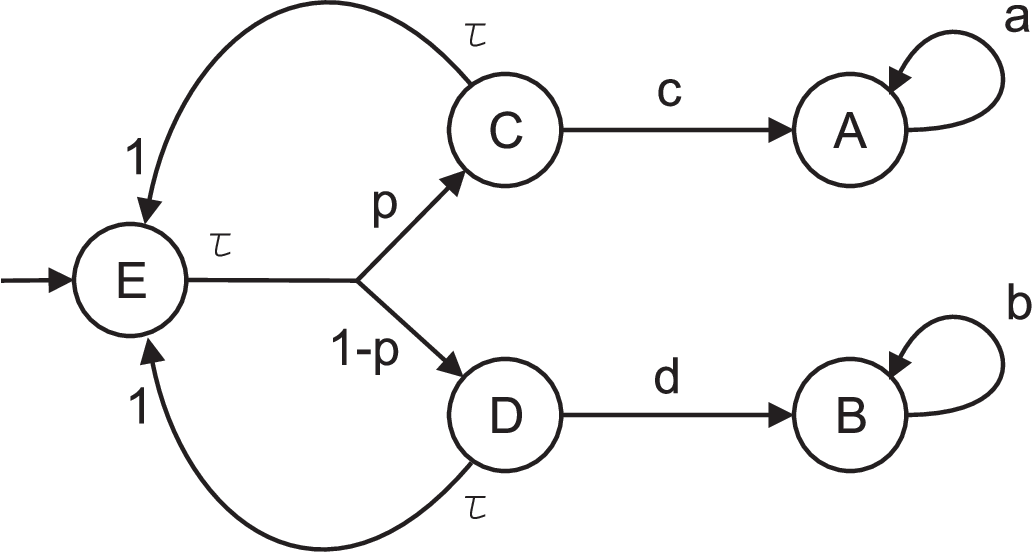}} \qquad
  \subfloat[E non-trivially \& na\"ively vanishing]{\label{fig:E_van}\includegraphics[width=5cm]{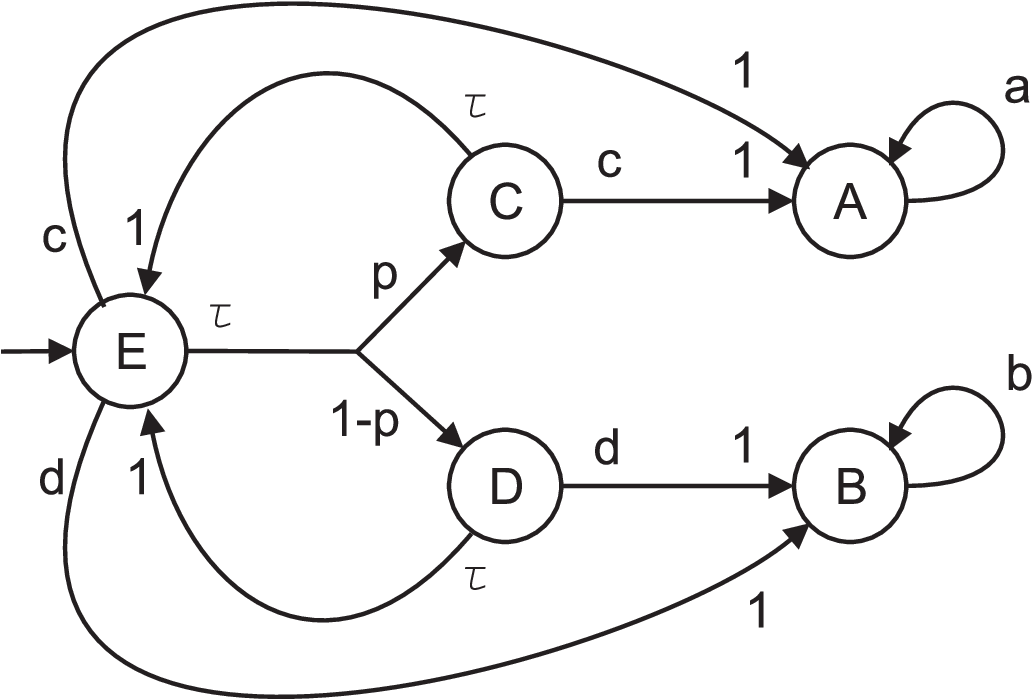}} \\
  \subfloat[E trivially \& nn-vanishing]{\label{fig:E_nn_van}\includegraphics[width=5cm]{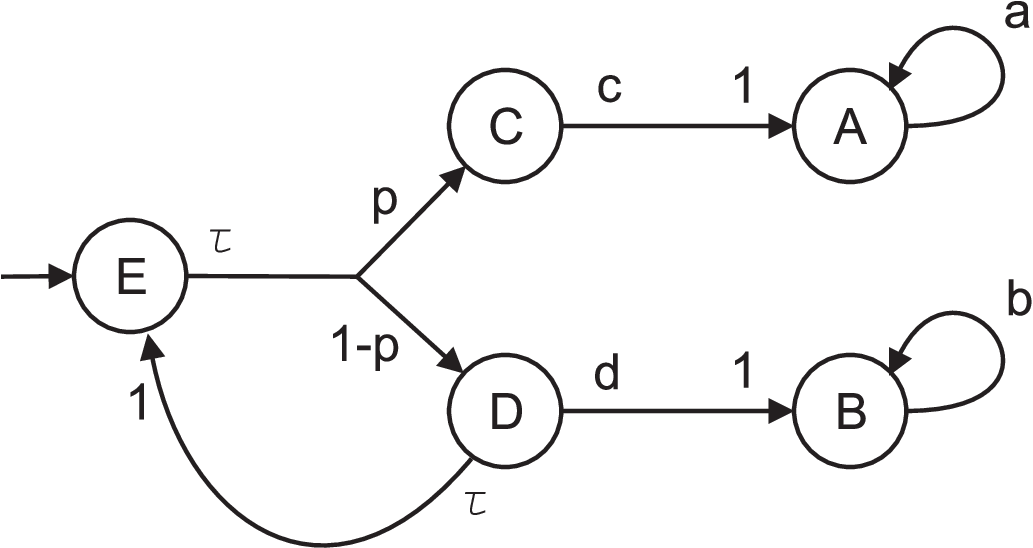}} \qquad
  \subfloat[E non-trivially \& nn-vanishing]{\label{fig:E_nn_van_nt}\includegraphics[width=5cm]{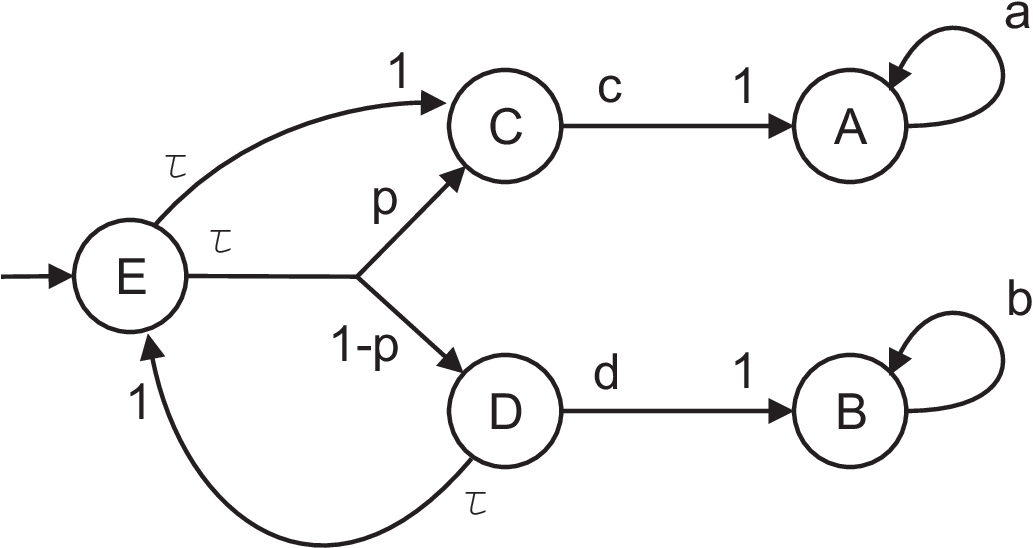}}
  \caption{Examples of vanishing states
}
  \label{fig:notions_of_vanishing}
\end{figure}

\begin{definition}[Elimination of vanishing states]
\label{def:elim}
Let $P=(S, Act, \rightarrow, \emptyset, s_0)$ be a PA. 
Let $s\in S$ be a vanishing state and let $s\stackrel{\tau}{\rightarrow}_{van}\nu$ be the only transition
emanating from $s$ in the vanishing representation $P'=P_{(s,\nu)}=(S, Act, \rightarrow_{van}, \emptyset, s_0)$.
The elimination of $s$ is defined by two steps:
\begin{enumerate}
  \item Rescaling (cf.~Remark \ref{rem:rescale}):
                    $$\rightarrow_{res}=\begin{cases}\rightarrow_{van} \setminus \{ (s,\tau,\nu)\} & \text{ if }\nu = \De_s \\
                                                    (\rightarrow_{van} \setminus \{ (s,\tau,\nu)\}) \uplus \{ (s,\tau,\frac{1}{1-\nu(s)}(\nu-s))\} & \text{ otherwise} \end{cases}$$
  \item
  Elimination (only performed if after rescaling a transition $s\stackrel{\tau}{\rightarrow}_{res}\nu_{res}$ remains):
      $$P^{\widehat{s}}=\begin{cases}(S\setminus \{s\}, Act, \rightarrow_{el}, \emptyset, s_0) & \text{ if } s \neq s_0 \\
                                  ((S\setminus \{s_0\})\uplus \{s_0^{\circ}\}, Act, \rightarrow_{el}\uplus \{ (s_0^{\circ}, \tau, \nu_{res}) \}, \emptyset, s_0^{\circ}) & \text{ if } s = s_0 \text{ and } \\
                                                                                                                                  & \exists t\rightarrow_{res}\gamma: s_0\in Supp(\gamma) \\
                                  P' & \text{otherwise}
                                  \end{cases}$$
\end{enumerate}
where
$\rightarrow_{el}:=\setcond{ (t,\alpha,\mu')}{ t\stackrel{\alpha}{\rightarrow}_{res}\mu, t\in S\setminus \{s\}, \mu':=\mu_{s\rightarrow \nu_{res}} }$.
Here $\mu_{s\rightarrow \nu_{res}}$ denotes the replacement of every occurrence of $s$ by the corresponding distribution $\nu_{res}$:

Without loss of generality let $\mu$ be of the form $\mu:=c_s \De_s \oplus ( \oplus_{i \in I, s_i\neq s} c_i\De_{s_i})$ and $\nu$ be of the 
form  $\nu=\oplus_{j \in J, s_j \neq s}d_j\De_{s_j}$.
Then we define $\mu_{s\rightarrow \nu}:=c_s(\oplus_{j \in J, s_j \neq s}d_j\De_{s_j}) \oplus ( \oplus_{i \in I, s_i\neq s} c_i\De_{s_i})$.

\end{definition}

We omit the $(s,\tau,\De_s)$ transition from the set of transitions for the purpose of minimality of the resulting PA. 
One could also just add the loop case to the case where $P'$ is not changed. Note that even when loops are removed, all
information about the MA may be safely recovered.
Such a state without loop is a deadlock in PA and no longer vanishing according to our definition (note that there cannot be any other competing transition,
as we then could not get the vanishing representation with the $\tau$ loop). Looking back to the MA setting it is clear that
$s$ must be an unstable state as it does not have the $\chi(0)$ transition.
We will see later when we describe the decision algorithm that we basically treat \emph{all} nn-vanishing
states as if they were starting states, i.e.~we consider them as transient copies.

\begin{example}
To explain Definition \ref{def:elim} we give the following examples.
The first example is the most common one (cf.~Fig.~\ref{fig:case1}): The vanishing state $s$ is neither the initial state nor does it 
have a probability-one-self-loop.
Therefore the elimination is straightforward: Redirect all incoming arcs according to the vanishing representation (cf.~Fig.~\ref{fig:case1_elim}).
The next example is the probability-one-self-loop case (cf.~Fig.~\ref{fig:case2}): It does not matter whether the vanishing state $s$ is the initial 
state or not, the self-loop is removed by the rescaling operation
(cf.~\ref{fig:case2_elim}). In the third example we have a vanishing initial state with incoming transition(s) (cf.~Fig.~\ref{fig:case3}). We add a copy $s_0^{\circ}$ of the initial state
and eliminate the \emph{old} initial state $s_0$ (cf.~Fig.~\ref{fig:case3_elim}). 
Note also that when $s$ is a vanishing initial state 
but it has no incoming transitions, then nothing is changed (without a figure).
\end{example}

\begin{figure}
  \centering
  \subfloat[Case 1]{\label{fig:case1}\includegraphics[width=3cm]{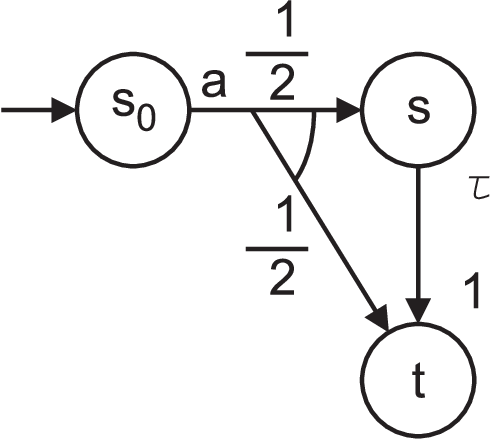}} \qquad \subfloat[Case 1 eliminated]{\label{fig:case1_elim}\includegraphics[width=3cm]{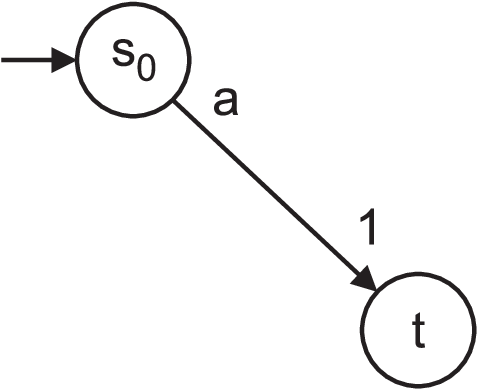}} \\
  \subfloat[Case 2]{\label{fig:case2}\includegraphics[width=3cm]{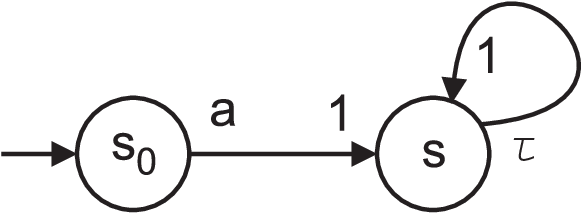}} \qquad \subfloat[Case 2 eliminated]{\label{fig:case2_elim}\includegraphics[width=3cm]{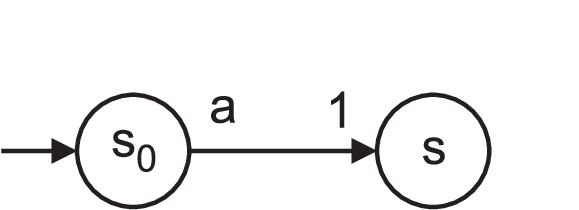}} \\
  \subfloat[Case 3]{\label{fig:case3}\includegraphics[width=3cm]{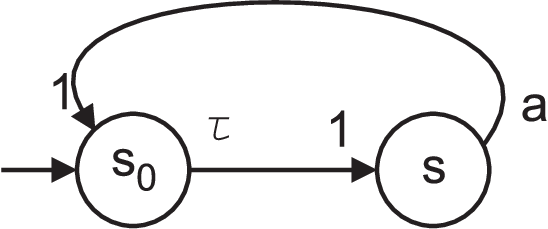}} \qquad \subfloat[Case 3 eliminated]{\label{fig:case3_elim}\includegraphics[width=3cm]{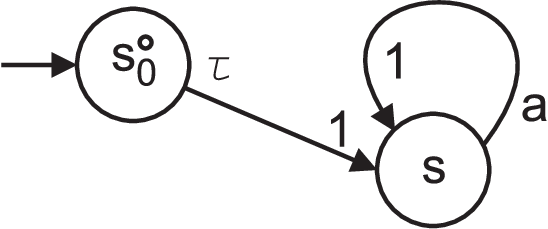}} \\
  \caption{Different cases of eliminations}
  \label{fig:cases_for_elimination}
\end{figure}

\begin{lemma}[Elimination does not destroy weak bisimilarity]
\label{lemma:1}
For every vanishing state $s$ it holds that $P\approx P'^{\widehat{s}}$
\end{lemma}

\begin{proof}
By definition $P\approx P'$. With the same arguments as in \cite{lics:10} (proof of Thm.~7) it follows that $P'\approx P'^{\widehat{s}}$.
So by transitivity of $\approx$ the claim follows.
\end{proof}

The following lemma helps to understand the difference between na\"ively vanishing and nn-vanishing states.
\begin{lemma}
\label{at_least_two_states}
For every vanishing representation $(s,\tau,\mu)$ that renders
state $s$ nn-vanishing there must be at least two 
distinct states $t_1,t_2 \in Supp(\mu)$ such that $t_1 \not \approx_\Delta s$, $t_2 \not \approx_\Delta s$ and $t_1 \not \approx_\Delta t_2$.
\end{lemma}

\begin{proof}
Let $(s,\tau,\mu)$ be the vanishing representation and assume that 
there is only \emph{one} state $t\in Supp(\mu)$ such that $t \not \approx_\Delta s$.
Without loss of generality, we may work on the quotient with respect to $\approx_\Delta$. Further we may assume that $\mu$ is rescaled, i.e.~$s\notin Supp(\mu)$\footnote{For otherwise pretend that $s$ is the initial state and eliminate it, i.e.~replace it by a transient copy. By Lemma \ref{lemma:1} it is clear that bisimilarity is not lost by this operation.}.
Therefore we may assume that $\mu=\Delta_t$.
Thus we obtain the vanishing representation $s\stackrel{\tau}{\rightarrow}\Delta_t$ from which it would follow that $s\approx_\Delta t$, which is a contradiction.
\end{proof}


\begin{lemma}
\label{lemma:elim_naively_van}
It is not always possible to eliminate all na\"ively vanishing states, but 
by elimination it is always possible to reach an automaton without na\"ively vanishing states.
\end{lemma}

\begin{proof}
A trivial example is given in Fig.~\ref{fig:not_nn_van}. Here both $s$ and $s'$ are na\"ively vanishing.
But only either $s$ or $s'$ may be eliminated (it is easy to see that after the first elimination, 
the other state is no longer
vanishing but has become tangible). 
\end{proof}

\begin{figure}
  \centering
  \includegraphics[width=5cm]{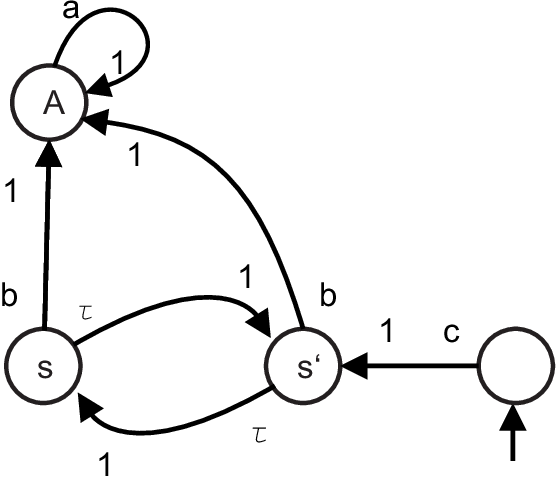}\\
  \caption{Two na\"ively vanishing states}
  \label{fig:not_nn_van}
\end{figure}


\begin{lemma}
\label{lem:all_nn_eliminated}
All nn-vanishing states 
can be eliminated
(except for a nn-vanishing initial state which can only be made transient).
\end{lemma}

\begin{proof}
We show that the situation described in the proof of
Lemma~\ref{lemma:elim_naively_van}, cannot occur for nn-vanishing states.
I.e.\ we show that the elimination of a nn-vanishing state $s$ may not
cause another nn-vanishing state $s'$ to lose its property of being 
nn-vanishing.
Again we work on the quotient with respect to $\approx_\Delta$.
Assume that $s$ is nn-vanishing with vanishing representation $(s,\tau,\mu)$.
Let $s'$ be another nn-vanishing state with vanishing representation $(s',\tau,\gamma)$.
In the case $s \not\in Supp(\gamma)$ the elimination of $s$ will not affect
$s'$.
In the case $s \in Supp(\gamma)$ the elimination of $s$ will cause $s$ to be
replaced by $\mu$ in the vanishing representation of $s'$.
We denote the resulting vanishing representation of $s'$ as 
$(s',\tau,\gamma')$\footnote{The resulting distribution $\gamma'$ is still bisimilar to the original one because of Lemma 4.}.
Since according to Lemma~\ref{at_least_two_states} we know that $Supp(\mu)$
still contains at least two states $t_1$ and $t_2$ such that $t_1 \not\approx_\Delta t_2$,
we conclude that $\gamma'$ contains at least one state ($t_1$ or $t_2$) that is not
in the $\approx_\Delta$ relation with $s'$.
Thus, after the elimination of state $s$, state $s'$ is still nn-vanishing.
\end{proof}

Directly from the proof of Lemma \ref{lem:all_nn_eliminated} we may deduce:
\begin{corollary}
\label{cor:nn-van-property-kept}
The property of a state $s$ being nn-vanishing is not destroyed by other states being eliminated.
\end{corollary}

\begin{example}
\label{ex:nn-van_rep}
This example shows that it is not enough to consider only strong emanating transitions when
searching for vanishing representations such that non-bisimilar states are reached.
Assume that $p\in (0,1)$.
We start with the automaton in Fig.~\ref{fig10:EF_nn}. Clearly both states $E$ and $F$ are trivially vanishing.
Considering only strong transitions we see that $E$ is nn-vanishing (vanishing representation $P_{(E,p\Delta_C\oplus (1-p)\Delta_D)}$), while 
$F$ would be erroneously
detected as na\"ively vanishing with its vanishing representation $P_{(F,\Delta_E)}$.
After elimination of $E$ we obtain the automaton in Fig.~\ref{fig10:F_nn}.
As elimination leads to bisimilar results (Lemma \ref{lemma:1})
we see -- after possibly rescaling the transition emanating from $D$ leading to Fig.~\ref{fig10:F_nn_rep} -- 
that also 
$F$ must be nn-vanishing with vanishing
representation $P_{(F,p\Delta_C\oplus (1-p)\Delta_D)}$.
\end{example}

\begin{figure}
  \centering
  \subfloat[E and F trivially vanishing]{\label{fig10:EF_nn}\includegraphics[width=7cm]{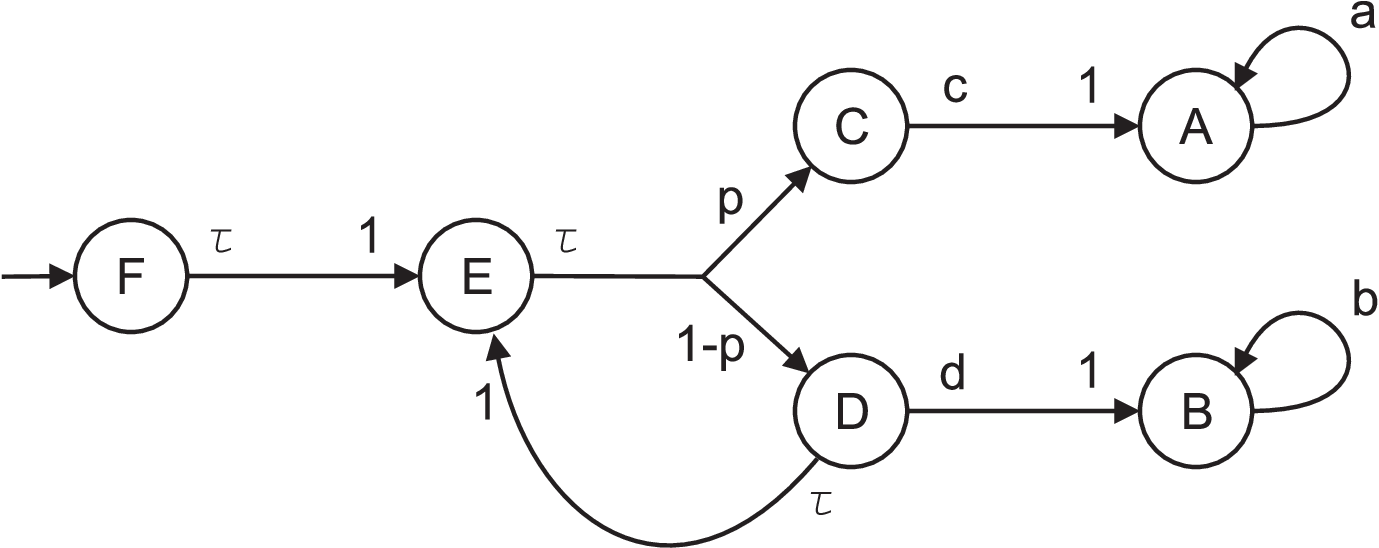}} \\
  \subfloat[E eliminated]{\label{fig10:F_nn}\includegraphics[width=5.5cm]{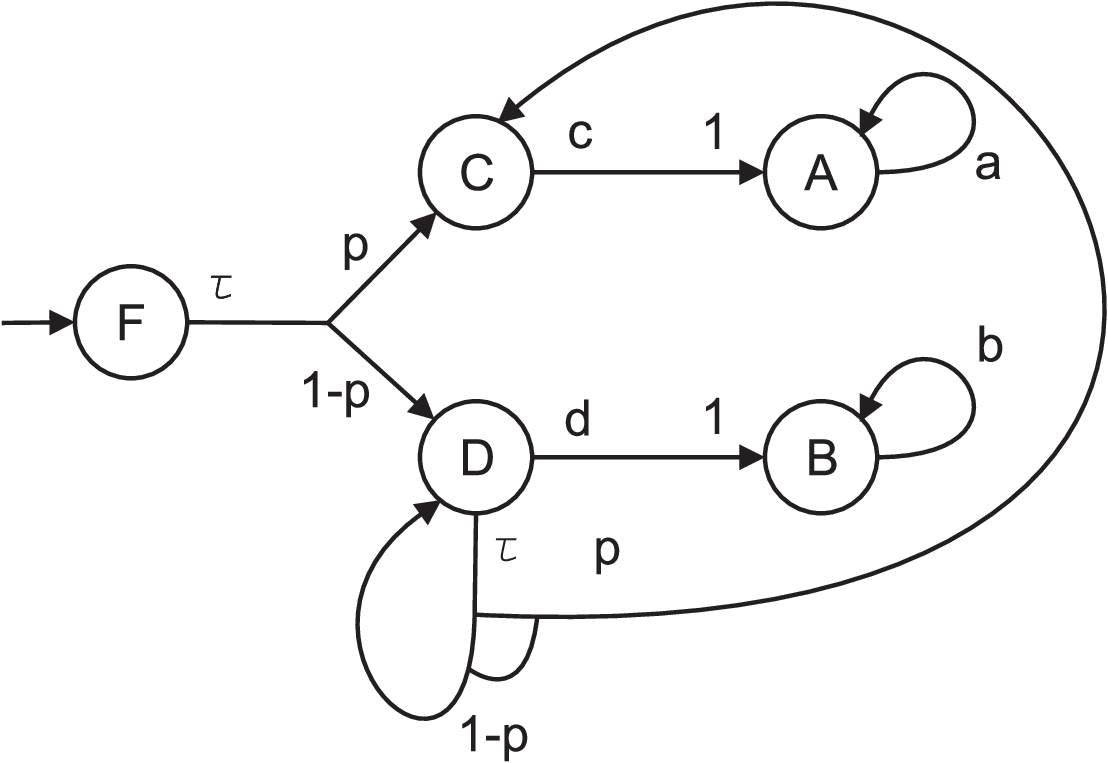}} \qquad
  \subfloat[D rescaled]{\label{fig10:F_nn_rep}\includegraphics[width=5cm]{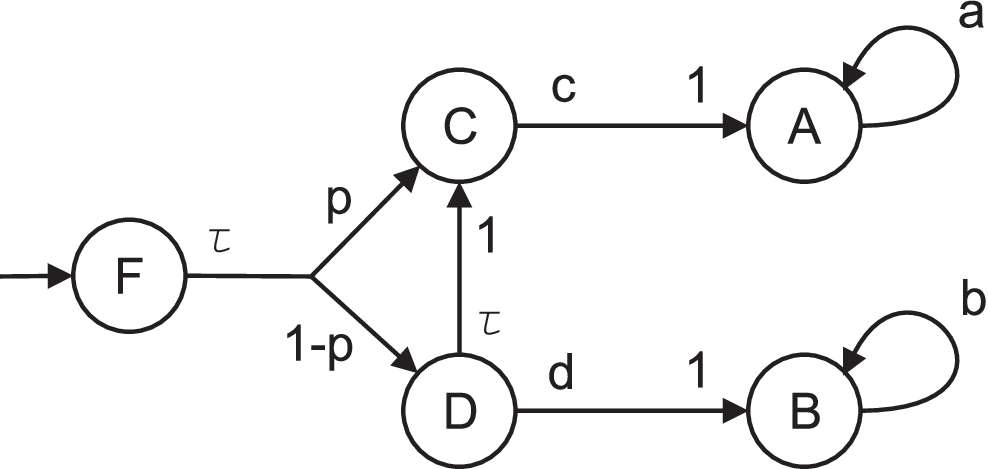}}
  \caption{Examples of nn-vanishing states}
  \label{fig:notions_of_vanishing2}
\end{figure}

We saw that all nn-vanishing states of an automaton may be eliminated, but not necessarily all na\"ively vanishing states. That is why 
we introduce $m\leq n$ in the following definition.
\begin{definition}
Let $P=(S, Act, \rightarrow, \emptyset, s_0)$. Let $S^v=\{s^v_1,\ldots,s^v_n\}$ be the set of vanishing states.
Denote by $\widehat{P}$ the \emph{complete elimination} of $P$, i.e.~
$\widehat{P}:=(\ldots(P'^{\widehat{s^v_{i_1}}})'^{\widehat{s^v_{i_2}}}\ldots)'^{\widehat{s^v_{i_m}}}$, $m\leq n$, $i_j\in \{1,\ldots,n\}$ for all $j\in\{1,\ldots,m\}$,
such that $\widehat{P}$ contains no more vanishing states.
Let $\widehat{P}^{\somesymbol}$ denote the elimination of all nn-vanishing states.
\end{definition}

Using this definition, we can state the following important lemma:

\begin{lemma}[Complete elimination and bisimilarity]
\label{lemma:elimbisim}
For two PA $P_1$ and $P_2$ it holds: $P_1\approx P_2 \Leftrightarrow \widehat{P}_1\approx \widehat{P}_2 \Leftrightarrow \widehat{P}^{\somesymbol}_1\approx \widehat{P}^{\somesymbol}_2$
\end{lemma}

\begin{proof}
By Lemma \ref{lemma:1} we know that elimination preserves weak bisimilarity.
The following diagrams show this by the right arrows.
$$
\xymatrix{ P_1 \ar[d] \ar[r]^\approx & \widehat{P}^{\somesymbol}_1 \ar[d] \ar[r]^\approx & \widehat{P}_1 \ar[d] \\
            P_2  \ar[r]^\approx        & \widehat{P}^{\somesymbol}_2 \ar[r]^\approx & \widehat{P}_2
            }
$$
As soon as one of the down arrows is a weak bisimulation, by transitivity of weak bisimulation we immediately get that the other two
arrows are also weak bisimulations.
%
\end{proof}


\begin{remark}
In every example from \cite{lics:10}, elimination leads
to isomorphic automata (assuming that we replace vanishing initial
states by their vanishing representation).
\end{remark}

It is clear by Lemma~\ref{lemma:naive-is-weak} that $\widehat{P}_1 \approx_{\text{na\"ive}} \widehat{P}_2\Rightarrow \widehat{P}_1 \approx \widehat{P}_2$
and therefore (by Lemma \ref{lemma:elimbisim}) $P_1 \approx P_2$.
Now we try to understand why it is also the case that $P_1 \approx P_2 \Rightarrow \widehat{P}_1 \approx_{\text{na\"ive}} \widehat{P}_2$.
\section{Canonical vanishing representations and properties of vanishing states}
\label{subsec:canon_van_repres}

In this section we prove that every nn-vanishing state $s$ has a vanishing representation $(s,\tau,\mu)$ where $Supp(\mu)$ only consists of 
nn-tangible states and the weak transition $s\stackrel{\tau}{\Rightarrow}\mu$ is driven by a Dirac determinate scheduler (i.e. it is not a combined transition).



\begin{lemma}[Vanishing representations don't need nn-vanishing states]
\label{lemma:non-vanishing}
Every nn-vanishing state $s$ has a vanishing 
representation $(s,\tau,\mu')$ where $Supp(\mu')$ does not contain
any nn-vanishing state
(i.e.\ only nn-tangible states are in $Supp(\mu')$).
\end{lemma}

\begin{proof}
Let $s$ be nn-vanishing with vanishing representation $(s,\tau,\mu)$.
We assume that $s \not\in Supp(\mu)$, otherwise rescale.
If $Supp(\mu)$ does not contain any nn-vanishing state we are done.
Otherwise we perform ``successive'' eliminations (and possibly rescalings)
of the nn-vanishing states in $\mu$ until all nn-vanishing states
are eliminated.
In detail: Replace all nn-vanishing states in $Supp(\mu)$ by their
vanishing representations, leading to a new vanishing representation
$(s, \tau, \mu')$ of state $s$.
If $s \in Supp(\mu')$ then rescale.
In case $Supp(\mu')$ contains only nn-tangible states we are finished.
Otherwise set $\mu := \mu'$ and perform another round of elimination / rescaling.
The finally resulting vanishing representation $(s,\tau,\mu')$ has only nn-tangible states.
The procedure terminates according to Lemma \ref{lem:all_nn_eliminated}.
\end{proof}

\begin{theorem}[nn-vanishing states correspond to ``real'' distributions]
\label{thm:nn-def-reformulation}
Let $P = (S, Act, T, \emptyset, s_0)$ be a PA.
A state $s\in S$ is nn-vanishing iff 
there exists a distribution $\mu$ such that
$\Delta_s\approx \mu$ but $\exists t\in Supp(\mu)$ such that $s\not \approx_\Delta t$.
\end{theorem}

\begin{proof}
$\Rightarrow$ Assume that $s$ is nn-vanishing, then we may use the vanishing representation.
Let $(s,\tau,\mu')$ be the vanishing representation of $s$.
Then we have both $\Delta_s \approx_\Delta \mu'$ (which follows directly from the definition of weak bisimilarity) and
$\exists t\in Supp(\mu')$ such that $s\not \approx_\Delta t$.
Therefore we can use $\mu = \mu'$ to satisfy the right hand side of the Theorem.

$\Leftarrow$ Without loss of generality we may work on quotients with respect to $\approx_\Delta$.
Assume that on the quotient it holds that $\mu=\oplus_{i\in \{0,\ldots,n\}}d_i \Delta_{t_i}$ for some $n\in\mathbb{N}$ ($t_i\neq t_j$ whenever $i\neq j$) and
without loss of generality we assume that $t^0 \not \approx_\Delta s$.
Firstly, we show that there must be a weak combined transition $s\stackrel{\tau}{\Rightarrow}_C \gamma$ 
with $\exists x\in Supp(\gamma):x \not \approx_\Delta s$ and $\Delta_s\approx \gamma$:
From $\Delta_s\approx \mu$ we get by \cite[Lemma~11]{avacsreport} a transition 
$s\stackrel{\tau}{\Rightarrow}_C\gamma=\oplus_{i\in \{0,\ldots,n\}}\gamma_i$
such that for all $i\in \{0,\ldots,n\}$ we get $d_i\Delta_{t_i}\approx \gamma_i$.
First note that by \cite[Lemma~9]{avacsreport} we get $\Delta_s\approx \gamma$, 
as from $d_i\Delta_{t_i}\approx \gamma_i$ it follows that $\mu=\oplus_{i\in \{0,\ldots,n\}}d_i \Delta_{t_i}\approx \oplus_{i\in \{0,\ldots,n\}}\gamma_i=\gamma$
and we have $\Delta_s\approx \mu$ as a precondition. So also $\Delta_s\approx \gamma$.
For $(s,\tau,\gamma)$ being a vanishing representation and $s$ being nn-vanishing, it remains to show that there exists some $x\in Supp(\gamma)$ such that $s\not \approx_\Delta x$.
Assume that this is not the case, i.e.~$\forall x\in Supp(\gamma):x\approx_\Delta s$.
Now assume that $\gamma_i=\oplus_{j\in J_i}b_{ij}\Delta_{x_{ij}}$ ($x_{ij}\neq x_{ik}$ whenever $j\neq k$).
But then we have 
$$d_{0}\Delta_s\not \approx d_{0}\Delta_{t^0} \approx \underbrace{\oplus_{j\in J_{0}}b_{0j}\Delta_{x_{0j}}}_{\gamma_0}\approx d_{0}\Delta_s $$
which is a contradiction
(the rightmost $\approx$ follows
from $\sum_{j\in J_0}b_{0j}=d_0$ and that for all $x_{0j}$ we have $x_{0j}\approx_\Delta s$,
therefore $\gamma_0\approx d_0\Delta_s$).
So we conclude that there must be a $x\in Supp(\gamma)$ such that $x\not\approx_\Delta s$.
Therefore we have a transition $s\stackrel{\tau}{\Rightarrow}_C\gamma$ with
$\Delta_s \approx \gamma$ and $\exists x \in Supp(\gamma): s\not \approx_\Delta x$.

Secondly, we have to show that $(s,\gamma)$ is a vanishing representation, i.e.~that $s \approx_\Delta t$ when comparing $\replace{P}{\replacement{s}{t}}$
and $\locallyChangedAut{s}{\sd}$.
So we have to show that when only using the transition $(s,\tau,\gamma)$ all other transitions emanating from $t$ can be mimicked and thus omitted.
Now let $t\stackrel{a}{\rightarrow}\rho$ be an arbitrary transition from $\replace{P}{\replacement{s}{t}}$. 
By $\Delta_t\approx \gamma$
we get directly from Definition \ref{def:weakbisim}
a hypertransition $\gamma\stackrel{a}{\Rightarrow}_C\rho'$
such that $\rho\approx \rho'$.
It remains to show that $\gamma\stackrel{a}{\Rightarrow}_C\rho'$ is possible in $\locallyChangedAut{s}{\sd}$ (i.e.~without using the transition $t\stackrel{a}{\rightarrow}\rho$).
If there was a transition $\gamma \stackrel{\tau}{\Rightarrow}_C \Delta_t$, this 
would be a contradiction to the existence of $\exists x\in Supp(\gamma):x \not \approx_\Delta t$
because all states in $Supp(\gamma)$ would be bisimilar to $t$ due to the loop $t\stackrel{\tau}{\Rightarrow}_C \gamma \stackrel{\tau}{\Rightarrow}_C \Delta_t$.
Therefore we know that such a transition does not exist\footnote{So the probability for returning to $s$ is strictly smaller than one.}.
That means we may successively substitute every occurrence of $t\stackrel{a}{\rightarrow}\rho$ by the 
transition $t\stackrel{\tau}{\Rightarrow}_C \gamma \stackrel{a}{\Rightarrow}_C \rho'$. 
Note that such a substitution clearly does not change bisimilarity, as the substituted distribution is bisimilar (by \cite[Lemma~9]{avacsreport}
distributions containing those different subdistributions will still be bisimilar).
The successive substitutions make the probability of choosing $t\stackrel{a}{\rightarrow}\rho$ tend to zero, i.e.~this transition is not needed.
As our example transition was general (and not equal to $t\Rightarrow \gamma$), we see that state $t$ must be nn-vanishing. 
\end{proof}

\begin{example}
A basic example of this kind of substitution used in the proof is given in Fig.~\ref{fig:example_substitution}.
\begin{figure}
  \centering
  \includegraphics[width=5cm]{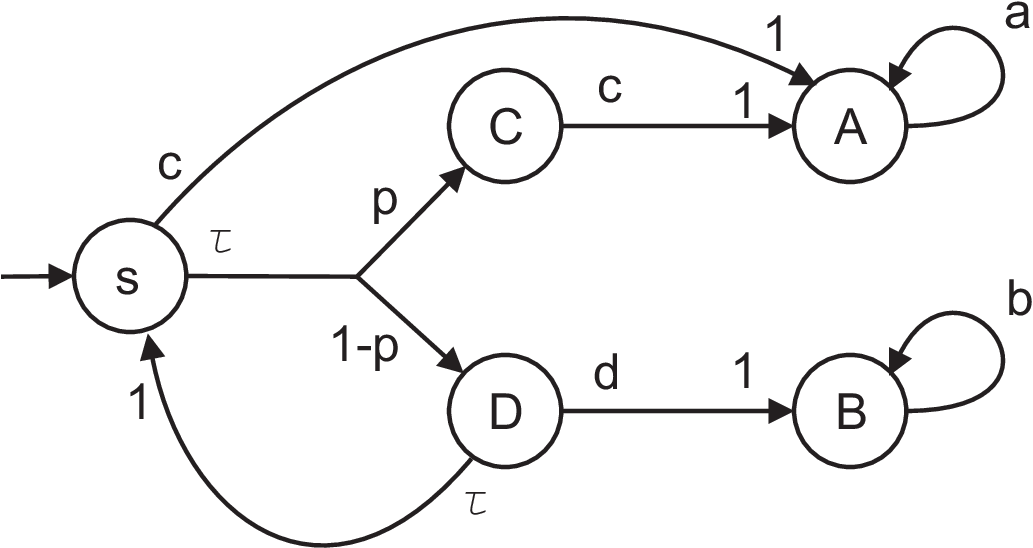}\\
  \caption{Example for the substitution of transitions
}
  \label{fig:example_substitution}
\end{figure}
Let $(s,p\Delta_C\oplus (1-p)\Delta_D)$ be the candidate for the vanishing representation. We have to show that
$s \stackrel{c}{\rightarrow}\Delta_A$ is superfluous.
According to the construction of the previous proof we get a hypertransition $p\Delta_C\oplus (1-p)\Delta_D\stackrel{c}{\Rightarrow}_C\Delta_A$.
We assume that this hypertransition is driven by the transitions $\Delta_C \stackrel{c}{\rightarrow}\Delta_A$ and $\Delta_D\stackrel{\tau}{\rightarrow}\Delta_s\stackrel{c}{\rightarrow}\Delta_A$.
This leads us to the sequence of transitions
$s \stackrel{\tau}{\rightarrow}(p\Delta_C\oplus (1-p)\Delta_D)\stackrel{\tau}{\Rightarrow}(p\Delta_C\oplus (1-p)\Delta_s)\stackrel{c}{\rightarrow} \Delta_A$,
where $s \stackrel{c}{\rightarrow}\Delta_A$ is taken with probability $(1-p)<1$, as assumed.
So in the next substitution step, we may use
$s \stackrel{\tau}{\rightarrow}(p\Delta_C\oplus (1-p)\Delta_D)\stackrel{\tau}{\Rightarrow}(p\Delta_C\oplus (1-p)\Delta_s)\stackrel{\tau}{\Rightarrow}(p(2-p)\Delta_C\oplus (1-p)^2\Delta_s)\stackrel{c}{\rightarrow} \Delta_A$,
that is we utilise $s \stackrel{c}{\rightarrow}\Delta_A$ only with probability $(1-p)^2$. Taking this to infinity means that the transition
$s \stackrel{c}{\rightarrow}\Delta_A$
 is indeed redundant, so
$(s,p\Delta_C\oplus (1-p)\Delta_D)$ is really a vanishing representation.
\end{example}

The following corollary states that all states within an equivalence class with respect to $\approx_\Delta$
are either nn-vanishing or not. 
\begin{corollary}
\label{lem:nn-van-on-classes}
Let $P=(S, Act, \rightarrow, \emptyset, s_0)$ be a PA. If $s\in S$ is nn-vanishing and $s\approx_\Delta t$
then $t$ is also nn-vanishing.
\end{corollary}
\begin{proof}
It follows from Thm.~\ref{thm:nn-def-reformulation} that
there exists a distribution $\mu$ such that
$\Delta_s\approx \mu$ but $\exists x\in Supp(\mu)$ with $s\not \approx_\Delta x$.
The assumption $s\approx_\Delta t$ yields, by transitivity of $\approx$ (cf.\ Thm.~1 in \cite{lics:10,avacsreport}),
that $\Delta_t\approx \mu$. Still $t\not \approx_\Delta x$, for otherwise $s\approx_\Delta t \approx_\Delta x$
which would be a contradiction to the nn-vanishing property of $s$. So, again by Thm.~\ref{thm:nn-def-reformulation},
we find that $t$ must be nn-vanishing, too.
\end{proof}

The following theorem states that the nn-vanishing states are the ``obstacle'' between ``weak bisimulation'' and ``na\"ive weak bisimulation''.
This theorem renders Thm.~2
in \cite{lics:10,avacsreport} more precisely.

\begin{theorem}
\label{th:main}
It holds that $P_1\approx P_2 \Leftrightarrow \widehat{P}^{\somesymbol}_1\approx_{\text{na\"ive}} \widehat{P}^{\somesymbol}_2$.
\end{theorem}

\begin{proof}
$\Leftarrow$ is immediate by Lemma~\ref{lemma:naive-is-weak} and Lemma \ref{lemma:elimbisim}.\\
%
$\Rightarrow$ From Lemma \ref{lemma:elimbisim} we already know $P_1\approx P_2 \Leftrightarrow \widehat{P}^{\somesymbol}_1\approx \widehat{P}^{\somesymbol}_2$.
So it remains to show that $\widehat{P}^{\somesymbol}_1\approx \widehat{P}^{\somesymbol}_2$ is already a na\"ive weak bisimulation. 
By the definition of weak bisimilarity (Definition \ref{def:weakbisim}) it follows that whenever $s\approx_\De t$ then
for every $s\stackrel{a}{\rightarrow}\mu$ we find $t\stackrel{a}{\Rightarrow}_C\gamma$
with $\mu \approx \gamma$ (and vice versa).
By Lemma \ref{lem:all_nn_eliminated} and Definition \ref{def:elim} it is clear that neither $\widehat{P}^{\somesymbol}_1$ nor $\widehat{P}^{\somesymbol}_2$
contain any nn-vanishing states (the only exception are, if present, nn-vanishing initial states, which then must be transient).
We have to show that $\mu\approx \gamma$ already coincide on classes, i.e.~the weak bisimulation is already na\"ive.
Assume that we split $\mu$ according to its support: $\mu=\oplus_{i\in I}c_i\De_{s_i}$, then
with Lemma 11 of \cite{avacsreport} we get a hypertransition $\gamma \Rightarrow_C \gamma'=\oplus_{i\in I}c_i\gamma'_i$ with $\De_{s_i}\approx \gamma'_i$.
By assumption $s_i$ cannot be nn-vanishing (only the initial state could be, but $s_i\neq s_0$ as it is a target state of some transition and $s_0$ is transient). 
Now it is clear that for all states $x \in Supp(\gamma'_i)$ it must hold that $x\approx_\Delta s_i$ -- for otherwise 
by Thm.~\ref{thm:nn-def-reformulation} $s_i$ would be nn-vanishing, which is a contradiction.
Summing up, we have $t\stackrel{a}{\Rightarrow}_C\gamma'$ 
and 
$\mu \equiv_{\approx_\De} \gamma'$.
We can use the same argumentation for $t\stackrel{a}{\rightarrow}\gamma$ to find $s\stackrel{a}{\Rightarrow}_C\mu'$ 
with 
$\gamma \equiv_{\approx_\De} \mu'$
and we conclude that
$\approx_\De$ is already a na\"ive weak bisimulation relation.

\end{proof}

The following lemma will be used in the proof of Lemma~\ref{lem:dd-is-enough},
which shows that it suffices to consider Dirac determinate schedulers when
trying to identify nn-vanishing states.
Note that, since Lemma~\ref{lem:unique} assumes that the PA at hand is a quotient
with respect to $\approx_\Delta$, it cannot contain any na\"ively vanishing states.

\begin{lemma}
\label{lem:unique}
Let $P=(S, Act, \rightarrow, \emptyset, s_0)$ a PA that is 
a quotient with respect to $\approx_\Delta$.
Let $s$ be a nn-vanishing state. For a vanishing representation $(s,\tau,\mu)$ where $Supp(\mu)$
contains only nn-tangible states
it holds that $\mu$ is uniquely defined, i.e.~there is no other vanishing representation $(s,\tau,\gamma)$
with $\mu \neq \gamma$
where $Supp(\gamma)$
contains only nn-tangible states.
(Such a vanishing representation $(s,\tau,\mu)$ is called canonical.)
\end{lemma}

\begin{proof}
Let $P'=(S, Act, \rightarrow, \emptyset, s_0)$ be a quotient with respect to $\approx_\Delta$, $s\in S$ nn-vanishing and assume that there exist different vanishing representations
$(s,\tau,\mu)$ and $(s,\tau,\gamma)$, where $Supp(\mu)$ and $Supp(\gamma)$ contain only nn-tangible states.
We now pretend that $s$ is the starting state, i.e.~consider the automaton $P=(S, Act, \rightarrow, \emptyset, s)$.
According to Thm.~\ref{th:main} it must hold that $\widehat{P_{(s,\mu)}}^\somesymbol \approx_{\text{na\"ive}} \widehat{P_{(s,\gamma)}}^\somesymbol$.
It is easy to see that $(s,\tau,\mu)$ and $(s,\tau,\gamma)$ are not changed by the elimination 
procedure\footnote{The transitions corresponding to the vanishing representations must already be 
rescaled as $s$ is nn-vanishing while all states in the support of $\gamma$ and $\mu$ are nn-tangible, 
so it holds that $s\notin Supp(\mu)$ and $s\notin Supp(\gamma)$. No state in the support of these 
distributions will have been eliminated.}.
As $\mu$ and $\gamma$ do not coincide and there are no other emanating transitions of state $s$, 
Thm.~1 of \cite{tacas:13} tells us that the corresponding Normal Form cannot have an emanating 
$\tau$ transition from $s$.
A contradiction to the nn-vanishing property of $s$. 
\end{proof}

The use of combined transitions $\newsymbol(s)$ for vanishing representations is not necessary,
as the following lemma shows:

\begin{lemma}[Considering Dirac determinate schedulers is sufficient to find nn-vanishing states]
\label{lem:dd-is-enough}
Every nn-vanishing state $s$ has a vanishing 
representation $(s,\tau,\mu)$ where $s\stackrel{\tau}{\Rightarrow}\mu$ is driven by a Dirac determinate scheduler.
\end{lemma}

\begin{proof}
Let $P^0=(S, Act, \rightarrow, \emptyset, s_0)$ be a PA. In the following we assume that we work on quotients with respect to $\approx_\Delta$.
We will show that any vanishing representation $(s,\tau,\mu)$
of nn-vanishing state $s$
can be transformed to a vanishing representation $(s,\tau,\mu')$
where $s\stackrel{\tau}{\Rightarrow}_C{\mu'}$
starts with a strong non-combined transition.                               
Let $\setcond{(s,\tau,\sd)}{\strongTransition{s}{\tau}{\sd}}=\{(s,\tau,\nu_1),\ldots,(s,\tau,\nu_n)\}$.
Assume now that we have a vanishing representation $(s,\tau,\mu)$ where the first strong step of 
$s\stackrel{\tau}{\rightarrow}_C \nu \stackrel{\tau}{\Rightarrow}_C\mu$ leads 
to a non-trivial combination $\nu=\sum_{i\in I} c_i \nu_i$ where $I\subseteq \{1,\ldots,n\}$ and $0<c_i<1$ for all $i\in I$.
Now we pretend for the moment that $s$ is the starting state, i.e.~we consider the automaton 
$P=(S, Act, \rightarrow, \emptyset, s)$. 
Fix for every nn-vanishing state $t$ a vanishing representation $(t,\tau,\mu_t)$ with support in nn-tangible states
(which is unique according to Lemma \ref{lem:unique}).
Now we perform two different kinds of eliminations:
\begin{enumerate}
  \item $\widehat{P}^{\somesymbol}$, i.e. the usual elimination from Def.~\ref{def:elim}
  \item $\widehat{P_{(s)}}^{\somesymbol}$,
        which means we make $s$ transient by elimination but keep all
        $\tau$-transitions
emanating from $s$, i.e.\ we do not move to its vanishing representation.
(All nn-vanishing states apart from $s$ are eliminated as usual.)
\end{enumerate}
By construction, all transitions in $\widehat{P}^{\somesymbol}$ and $\widehat{P_{(s)}}^{\somesymbol}$
lead to distributions whose support consists only of nn-tangible states (this will be denoted by 
the superscript ``nn-tang'').
As the common starting point for reaching eliminations $\widehat{P}^{\somesymbol}$ and $\widehat{P_{(s)}}^{\somesymbol}$
is the automaton $P$, we get
by Thm~\ref{th:main} that 
$\widehat{P}^{\somesymbol}\approx_{\text{na\"ive}} \widehat{P_{(s)}}^{\somesymbol}$,
so the resulting transition in $\widehat{P_{(s)}}^{\somesymbol}$ emanating from $s$ leading 
to nn-tangible states must still 
be a vanishing representation.
In $\widehat{P}^{\somesymbol}$ the transitions $s\stackrel{\tau}{\rightarrow} \nu_i$ are transformed to 
strong non-combined 
transitions $s\stackrel{\tau}{\rightarrow} \nu^{\texttt{nn-tang}}_i$.
In $\widehat{P_{(s)}}^{\somesymbol}$ the transition $s\stackrel{\tau}{\rightarrow}_C \nu$ is 
transformed to a strong (possibly combined)
transition $s\stackrel{\tau}{\rightarrow}_C \nu^{\texttt{nn-tang}}$.
Now we need the following reduction argument: If there are $i,j$ such that $\nu^{\texttt{nn-tang}}_i=\nu^{\texttt{nn-tang}}_j$,
then we can substitute $\nu_i$ by $\nu_j$ in $\nu=\sum_{i\in I} c_i \nu_i$ without losing the property of $(s,\tau,\nu)$
being the first step of
 a vanishing representation (by Thm.~\ref{thm:nn-def-reformulation}). 
(As an example we have in Fig.~\ref{fig:chosing_strong_nn_transitions} the distributions 
$\nu_1=\frac{2}{3}\Delta_x\oplus \frac{1}{3}\Delta_C$ and
$\nu_2=\frac{1}{3}\Delta_A\oplus\frac{2}{3}\Delta_y$, which leads to 
$\nu^{\texttt{nn-tang}}_1=\nu^{\texttt{nn-tang}}_2=\frac{1}{3}\Delta_A\oplus \frac{1}{3}\Delta_B\oplus\frac{1}{3}\Delta_C$ ---
therefore we may substitute $\nu_1$ by $\nu_2$ in $\nu=\sum_{i\in I} c_i \nu_i$). 
\begin{figure}
  \centering
  \includegraphics[width=5cm]{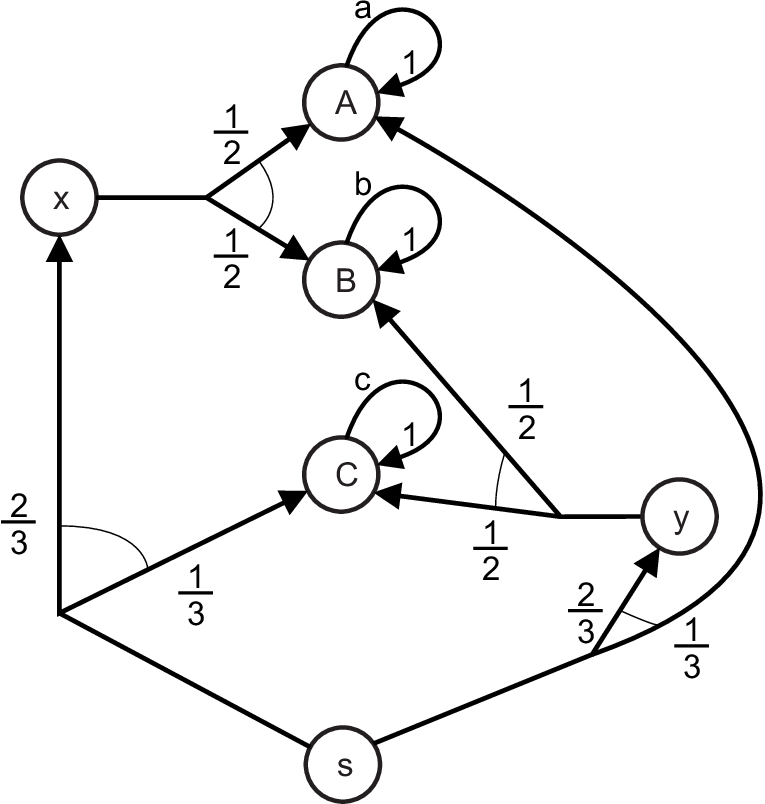}\\
  \caption{Coinciding $\nu^{\texttt{nn-tang}}_i$ and $\nu^{\texttt{nn-tang}}_j$
  }
  \label{fig:chosing_strong_nn_transitions}
\end{figure}
If after this reduction argument, the vanishing representation has a strong transition as first step,
we are done. Otherwise we may assume 
that $\nu=\sum_{i\in I^*} c^*_i \nu_i$ with $I^*\subseteq I$, where $\nu^{\texttt{nn-tang}}_i\neq \nu^{\texttt{nn-tang}}_j$ for
$i \neq j$. 
As for every nn-vanishing state $t$ a \emph{unique} vanishing representation
$(t,\tau,\mu_t)$ was chosen,
 it must hold that
$\nu^{\texttt{nn-tang}}=\sum_{i\in I^*} c^*_i \nu^{\texttt{nn-tang}}_i$ 
(with the same coefficents $c^*_i$ as above).
As $s$ is transient in $\widehat{P}^{\somesymbol}$ and $\widehat{P_{(s)}}^{\somesymbol}$
by definition,
all transitions emanating from $s$ in these eliminated automata must be rescaled.
Now we may apply Thm.~1 of \cite{tacas:13} to see that the corresponding
normal form
(whose set of transitions must be from the intersection
of the transition sets of
$\widehat{P}^{\somesymbol}$ and $\widehat{P_{(s)}}^{\somesymbol}$)
cannot have an emanating 
$\tau$ transition from
$s$\footnote{The only possible vanishing representation on the quotient would be $(s,\tau, \Delta_s)$, which clearly doesn't satisfy the nn-vanishing property.}.
This is a contradiction to the nn-vanishing property of $s$.
This means that the first step of a transition leading to a vanishing representation must be non-combined, 
i.e.\ some $c^*_i=1$ in the sum above.
As the same argument holds for the other 
nn-vanishing states, we see that with Lemma \ref{lemma:non-vanishing}
a vanishing representation $(s,\tau,\rho)$ for $s$ with some $t\in Supp(\rho): s\not\approx_\Delta t$ can be reached from $s$ by 
Dirac Determinate schedulers. 
%
%
\end{proof}

It can be shown that as long as there are no probability-one $\tau$-loops,
the elimination procedure is unique.
At the presence of such loops the elimination procedure is unique up to isomorphism \cite{schuster:11}.

\begin{corollary}
It holds that $P_1\approx P_2 \Leftrightarrow \widehat{P}_1\approx_{\text{na\"ive}} \widehat{P}_2$.
\end{corollary}

\begin{proof}
$\Leftarrow$ is immediate by Lemma~\ref{lemma:naive-is-weak} and Lemma \ref{lemma:elimbisim}.\\
$\Rightarrow$
Using Thm.~\ref{th:main} we know that after elimination of nn-vanishing states $\approx_\Delta$ is already
a n\"aive weak bisimulation, so it only remains to show that 
$\widehat{P}^*_1\approx_{\text{na\"ive}} \widehat{P}^*_2 \Leftrightarrow
\widehat{P}_1\approx_{\text{na\"ive}} \widehat{P}_2$.
This is true since eliminating any na\"ively vanishing states from 
$\widehat{P}^*_1$ or 
$\widehat{P}^*_2$ does not change the behaviour with respect to  
$\approx_\Delta$, so it is still a na\"ive weak bisimulation.
\end{proof}

\section{A partition refinement algorithm}
\label{sec:bisimalgo}



With Lemma \ref{lem:dd-is-enough}
we can find nn-tangible states and Th.~\ref{th:main} reduces the problem to na\"ive weak bisimulation. 
%
%
%
For the description of the partition refinement algorithm below we need the convex sets $S(x,a)\subseteq \BR^n$ introduced by \cite{segala:02}.
For details on how to calculate those sets we refer to \cite{segala:02}.
\begin{example}
For the MA given in Fig.~\ref{fig:nondet} 
the set $S(s_1,\tau)$ is given by the shaded triangle in Fig.~\ref{fig:triangle}.
This triangle encodes all distributions that are reachable via a (weak) combined $\tau$ transition starting from $s_1$.
\end{example}

\begin{figure}
  \centering
  \subfloat[Example automaton]{\label{fig:nondet}\includegraphics[width=3cm]{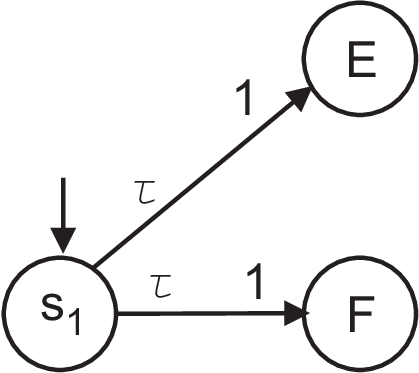}} \qquad
  \subfloat[$S(s_1,\tau)$]{\label{fig:triangle}\includegraphics[width=3cm]{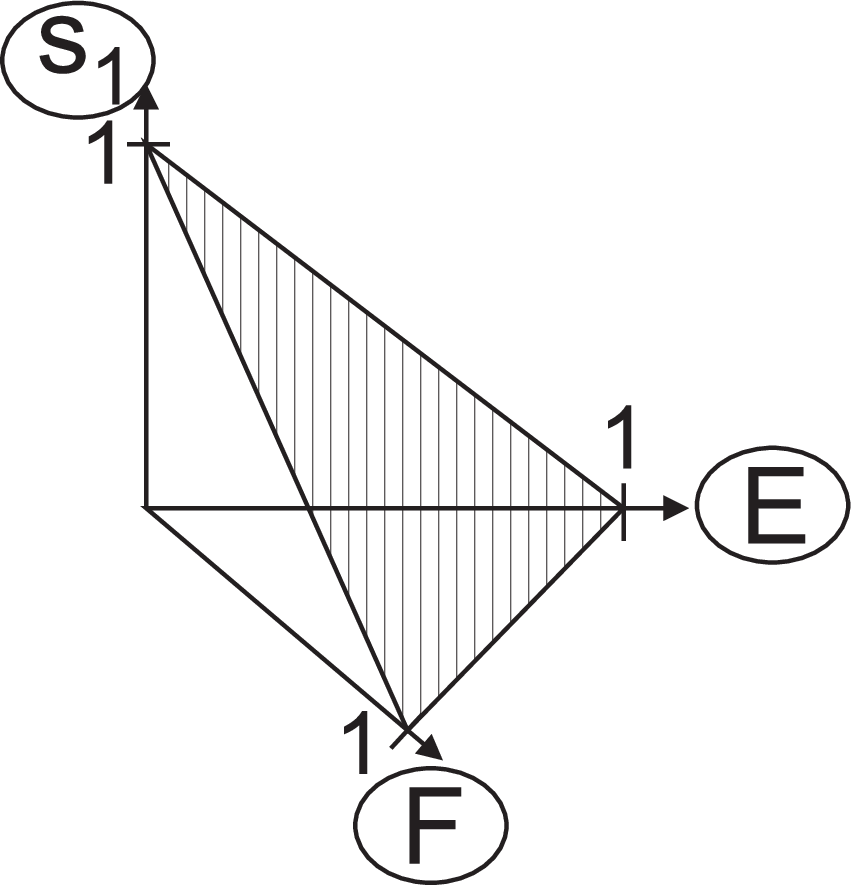}} 
  \caption{Example for reachable distributions}
  \label{fig:notions_of_vanishing3}
\end{figure}

%

\begin{remark}
\label{rem:convex}
It was shown in
\cite{segala:02} that each convex set is
the convex hull (CHull) of distributions 
(i.e.~points in $\BR^n$)
generated by Dirac determinate schedulers.
It is shown there that the extremal points (i.e.~generators) of the convex hull can be found by a linear program
and that the complexity of calculating the sets $S(s,a)$ is exponential for the weak case.
This is one of the reasons why
our algorithm also has exponential complexity
(see Sec.~\ref{subsec:Complexity}). 
\end{remark}



Given a set $S(s,a)\subseteq \BR^n$, it may be restricted, which is an essential ingredient of our decision algorithm.
\begin{definition}[Restriction of convex sets]
We define 
$$S(s,a)|_{x_i=0}:=\{(x_1,\ldots,x_n) \in S(s,a)| x_i = 0 \}$$
Multiple restrictions can also be realised. Especially, we define a restriction
to a set of ``nn-tangible'' states $S(s,a)|_{\TANGSTATES}$
 by requiring that $x_i=0$ for all nn-vanishing states $x_i$.
\end{definition}

Note that the restriction of
convex sets is still convex by definition. 

\begin{lemma}
\label{lem:restrict_elim}
Restriction of $S(s,a)$ to the set of nn-tangible states corresponds to
virtual elimination of nn-vanishing states.
\end{lemma}

\begin{proof}
In the elimination procedure we redirect transitions according to the vanishing representation. The vanishing state is no longer reached
and can be removed from the state space (or alternatively: its probability can be restricted to zero).
All other transitions that can be realised from a nn-vanishing state $s$ using other transitions than the one belonging to the vanishing 
representation can be weakly emulated by using the vanishing representation as a first step.
%
The only exception to keep in mind is when considering the sets $S(s,a)$
where $s$ itself is nn-vanishing. 
There normally $\Delta_s$ must be in $S(s,\tau)$ (as of course always $s\stackrel{\tau}{\Rightarrow}\Delta_s$). 
As $s$ is nn-vanishing, there must also be a vanishing
representation 
consisting of
nn-tangible states which uniquely identifies $s$ (cf.~Lemma \ref{lemma:non-vanishing} and Lemma \ref{lem:unique}).
Note that it holds that $S(s,\tau)=CHull(\Delta_s, S(s,\tau)|_{prob(s)=0})$ and therefore 
$S(s,\tau)$ is uniquely determined by the set $S(s,\tau)|_{prob(s)=0}$.
Similar considerations apply also for the case $a\neq \tau$: nn-vanishing states are identified by their vanishing representations consisting of nn-tangible states. 
\end{proof}

The last missing part for setting up the algorithm is to show how it can be
fitted
 to a partition refinement algorithm.
This can be seen in Fig.~\ref{fig:transient_rendering}. 
The underlying idea 
is to 
treat every nn-vanishing state in the elimination as if it were an
initial state (i.e.~eliminate it, but leave a transient copy in the transition system).
Note, however, that the algorithm does not perform a real elimination, but only a ``virtual'' elimination (by considering the restricted $S(s,a)$-sets).
But now it is clear by Thm.~\ref{th:main} that
the bisimulation problems we have to solve are na\"ive weak bisimulations that can be treated by the Segala/Cattani algorithm. 
For example, in Fig.~\ref{fig:transient_rendering} one trivially sees that $E$ and $F$ are na\"ively weakly bisimilar.
But after a splitting
occurred, we have to verify if all nn-vanishing states are still nn-vanishing with respect to the new partition. This justifies the iterative scheme
sketched in Fig.~\ref{fig:algo_schematic}.
\begin{figure}
  \centering
  \subfloat[E and F nn-vanishing]{\label{fig:EF_nn}\includegraphics[width=7cm]{EF_nn}} \\
  \subfloat[E eliminated (treating it as if it were a starting state)]{\label{fig:F_nn}\includegraphics[width=5.5cm]{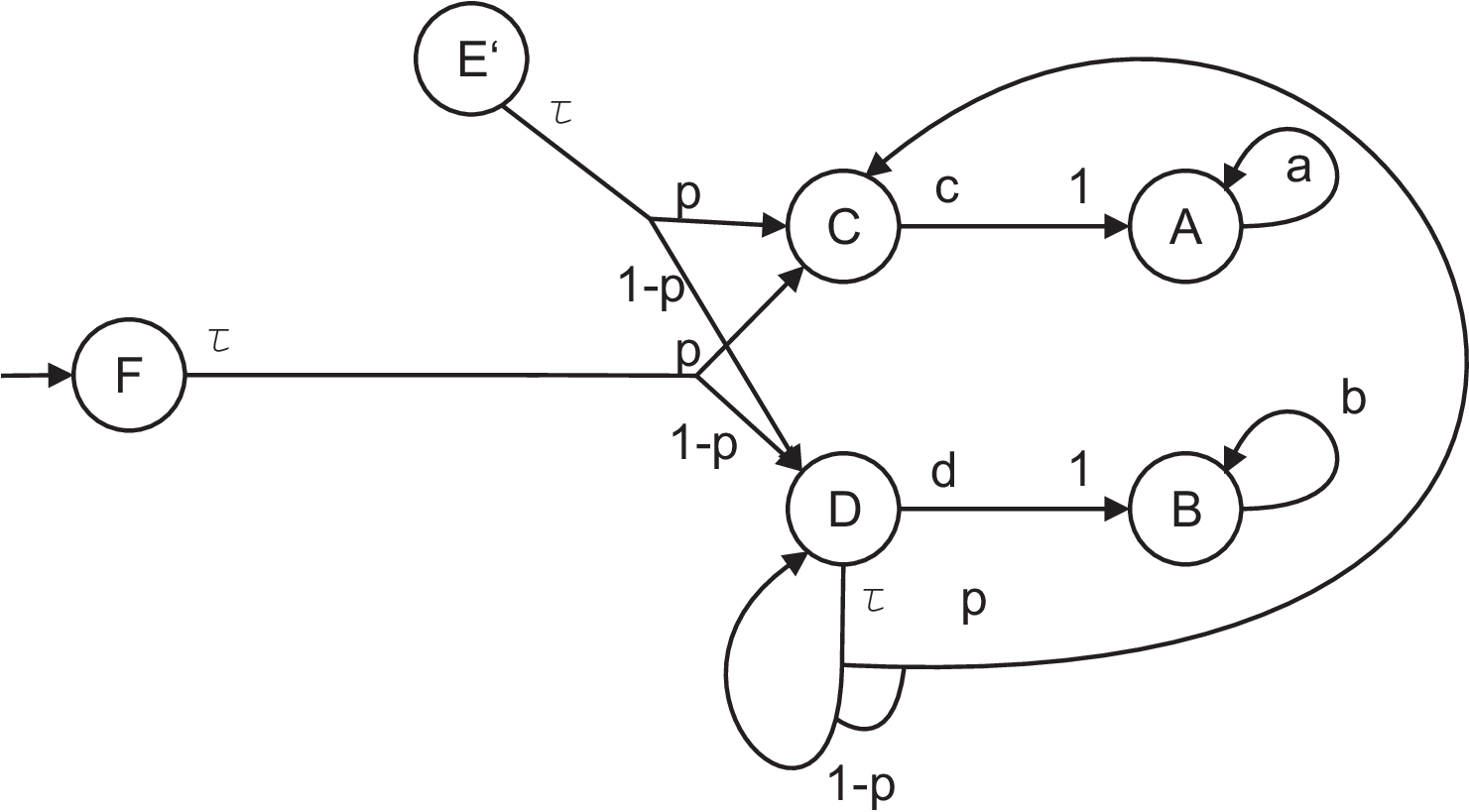}}
  \caption{Rendering nn-vanishing states transient}
  \label{fig:transient_rendering}
\end{figure}

So the proposed algorithm
for deciding whether two MA are weakly bisimilar
looks as follows:
\begin{enumerate}
  \item Start with the initial partition $\WW = \{S_1\uplus S_2\}$.
  \item For all states $s$ and actions $a$ calculate the convex sets $S(s, a)$ (cf.~\cite{segala:02}) 
         and for every (Dirac determinate) weak transition $(s,\tau,\nu)$ calculate $S_{\nu}(s,a)$ 
         ($S_{\nu}(s,a)$ denotes the convex set calculated for the PA $(P_1\uplus P_2)_{(s,\nu)}$, 
         that is the direct sum of automata $P_1$ and $P_2$ where we move to the vanishing
         representation $(s,\nu)$).
  \item Set $\TANGSTATES=\emptyset$.
  \item For all states $s$ that are not in $\TANGSTATES$
    \begin{itemize}
       \item Check whether $s$ can leave its equivalence class in $\WW$
by a (Dirac determinate) weak transition $(s,\tau,\nu)$
         such that (modulo $\WW$)\\
 $S(s,a)|_{\TANGSTATES}=S_{\nu}(s,a)|_{\TANGSTATES}$ for all $a \in Act$. 
         This is a \emph{vanishing representation of nn-vanishing state $s$ with respect to $\WW$}.
       \item If no vanishing representation with respect to the current partition $\WW$ can be found, then $s$ must be nn-tangible with respect to $\WW$.
         Add state $s$ to $\TANGSTATES$.
       \item Cross-check all other states if they also become nn-tangible as an effect of $s$ being nn-tangible.
    \end{itemize}
  \item Find a new splitter (in the sense of \cite{segala:02}) with respect to the current partition and the current set of nn-tangible states, 
         i.e.~a tuple $(C,a, \WW)$ which indicates that class $C$ needs to be refined w.r.t.~a weak $a$ transition.
  \item Refine the partition according to the splitter and
start next round at step 3.
\end{enumerate}

The algorithm is depicted in Alg.~\ref{alg1}.
By $DiracDet(s,\tau)$ we mean all distributions $\nu$ induced by $s\Rightarrow \nu$ by means of a Dirac determinate scheduler. 
It remains to define the ComputeInfo algorithm, FindWeakSplit algorithm and the Refine algorithm. As the Refine algorithm is standard, we omit
it from this paper. The routine ComputeInfo just calculates the convex sets $S(s,a)$ according to Remark \ref{rem:convex}. The routine FindWeakSplit given in Alg.~\ref{alg2} looks very much like the one given in \cite{segala:02} 
but it only ``sees'' nn-tangible states that 
are provided as an additional parameter to the routine.

\begin{lemma}
The algorithm in Alg.~\ref{alg1} calculates the coarsest partition with respect to $\approx_\Delta$.
\end{lemma}

\begin{proof}
The claim follows by the correctness of the na\"ive weak bisimulation algorithm given in \cite{segala:02}.
The only special case to consider is when a nn-vanishing and a nn-tangible state are detected in the \emph{same} class.
This case is not problematic due to the following reasoning: the nn-vanishing state can leave its class 
towards 
an equivalent
 distribution
(i.e.\ a distribution weakly bisimilar to the Dirac distribution on the
nn-vanishing state)
which consists of at least two other classes (cf.\ Lemma~\ref{at_least_two_states}).
In contrast, the nn-tangible state either cannot leave its class at all, or it can leave its class but thereby losing bisimilarity.
As the classes are refined and never merged, by the above reasoning nn-vanishing and nn-tangible states cannot be bisimilar and may always be split. 
Therefore the special case that an
nn-tangible state $s$ has $\Delta_s$ in $S(s,\tau)|_{\TANGSTATES}$, whereas -- due to restriction -- for an nn-vanishing state $t$, $\Delta_t$ 
is not in $S(s,\tau)|_{\TANGSTATES}$
is not problematic as it 
 allows for separating nn-vanishing from nn-tangible states.
Similar considerations apply to the case $a=\tau$, which we do not discuss explicitly.

Once the algorithm terminates, there is no class containing both nn-vanishing and nn-tangible states. 
When restricting to the nn-tangible fraction of successor states (which corresponds by Lemma~\ref{lem:restrict_elim} to elimination of nn-vanishing states), 
the states within one class are na\"ively weakly bisimilar (this follows 
from the algorithm given in \cite{segala:02}). Furthermore, our algorithm calculates for each nn-vanishing state the canonical vanishing representation consisting
of nn-tangible states only (cf.~Lemma~\ref{lemma:non-vanishing}).
Summing up, this means (by Thm.~\ref{th:main}) weak bisimilarity when considering both nn-tangible and nn-vanishing states (i.e.\ without restrictions).
\end{proof}

\begin{algorithm}                     
\caption{DecideWeakBisim}  
\label{alg1}                           
\begin{algorithmic}[1]                    
    \REQUIRE Two MA 
              as PA $P_1=(S_1, Act_1, \rightarrow_1, \emptyset, s_0)$, $P_2=(S_2, Act_2, \rightarrow_2, \emptyset, t_0)$
    \STATE $S=S_1\uplus S_2$, $\WW = \{S\}$, $Act = Act_1 \cup Act_2$
    \FOR{$s\in S$, $a \in Act$, $\nu \in \mathit{DiracDet}(s,\tau)$}
      \STATE $S(s,a)=\mathit{ComputeInfo}(s,a)$ on $(P_1\uplus P_2)$
      \STATE $S_{\nu}(s,a)=\mathit{ComputeInfo}(s,a)$ on $(P_1\uplus P_2)_{(s,\nu)}$
    \ENDFOR
    \WHILE{$\WW$ changes}
    \STATE $\TANGSTATES:=\emptyset$
    \WHILE{$\TANGSTATES$ changes} 
      \FOR{$s \in S\setminus \TANGSTATES$}
        \FOR{$\nu \in \mathit{DiracDet}(s,\tau)$ where $\exists x \in Supp(\nu): [x]_{\WW}\neq [s]_\WW$}  
          \IF{$\forall a \in Act : (S(s,a)|_{\TANGSTATES})/\WW = (S_{\nu}(s,a)|_{\TANGSTATES})/\WW$}
            \STATE vanishing representation found \textbf{break}
          \ENDIF
        \ENDFOR
        \IF{no vanishing representation found}
          \STATE $\TANGSTATES:=\TANGSTATES \cup \{s\}$
        \ENDIF
      \ENDFOR
    \ENDWHILE
    \STATE $(C,a,\WW)=\mathit{FindWeakSplit}(\TANGSTATES,\WW,S,Act,S(\cdot,\cdot))$
    \STATE $\WW=\mathit{Refine}(C,a,\WW)$
    \ENDWHILE
    \STATE $P_1\approx P_2$ iff $[s_0]_\WW = [t_0]_\WW$
\end{algorithmic}
\end{algorithm}

\begin{algorithm}                      
\caption{FindWeakSplit (Find weak bisimulation splitter)}  
\label{alg2}                           
\begin{algorithmic}[1]
  \REQUIRE nn-tangible
states $\TANGSTATES$, partition $\WW$, states $S$, actions $Act$, Info $S(\cdot,\cdot)$   
    \FOR{$C_i \in \WW$, $s,t \in C_i$, $a\in Act$}
      \IF{$(S(s,a)|_{\TANGSTATES})/\WW \neq (S(t,a)|_{\TANGSTATES})/\WW$}
         \STATE return $(C_i,a, \WW)$
      \ENDIF
    \ENDFOR
\end{algorithmic}
\end{algorithm}

\begin{remark}
The algorithm detects nn-vanishing states and finds their vanishing representations (regarding the most recent partition $\WW$).
Therefore, at the end of the algorithm,
we will be able to really (i.e.\ not virtually) eliminate
 all nn-vanishing states and reach a form where only nn-tangible states are present
(with the only exception of a nn-vanishing initial state,
 which can only be rendered transient)
and where for every equivalence class only one state is used. 
This can be regarded as a kind of normal form.
\end{remark}

\subsection{Example}
Suppose we are given the MA in Fig.~\ref{fig:v1} (there already transformed to a PA P) where $p,q\in (0,1)$.
This automaton can be seen as a condensed form of two separate automata (starting with $s_1$ and $t_1$, thus these are indicated as initial states), 
where states $A$ and $B$ 
have been identified (to keep things short -- if there were two copies of A and B: one for the left and one for the right automaton, they would be grouped in the course of the algorithm).
We want to show that $s_1\approx_\De t_1$.


\begin{figure}
  \centering
  \subfloat[Non-trivial example]{\label{fig:v1}\includegraphics[width=5cm]{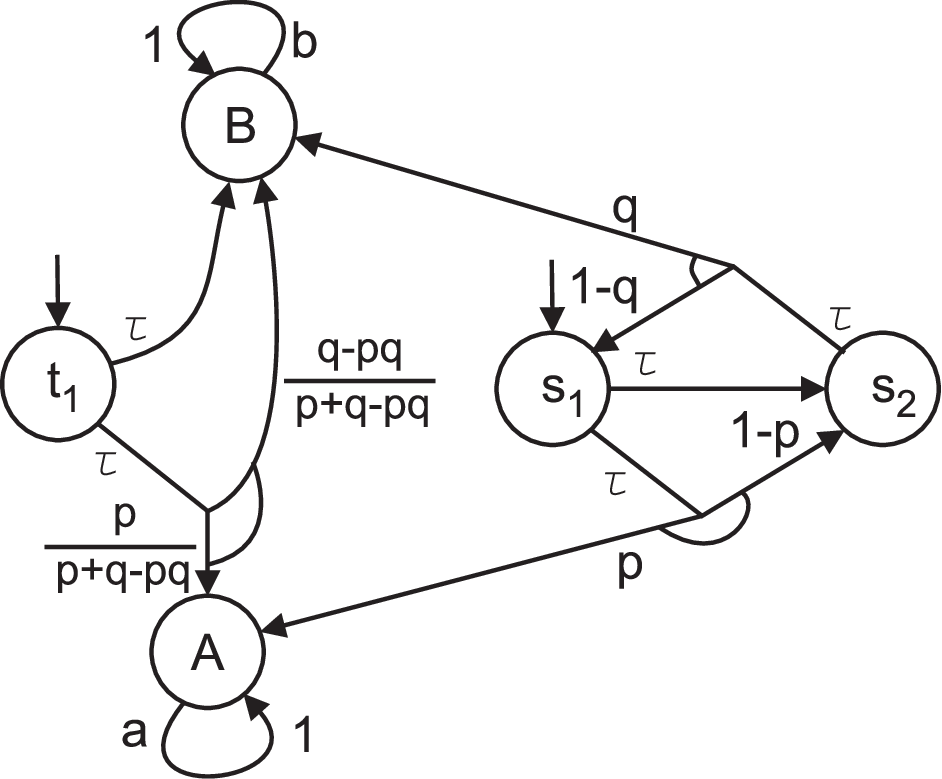}} \\
  \subfloat[$\WW_0$]{\label{fig:part0}\includegraphics[width=2.7cm]{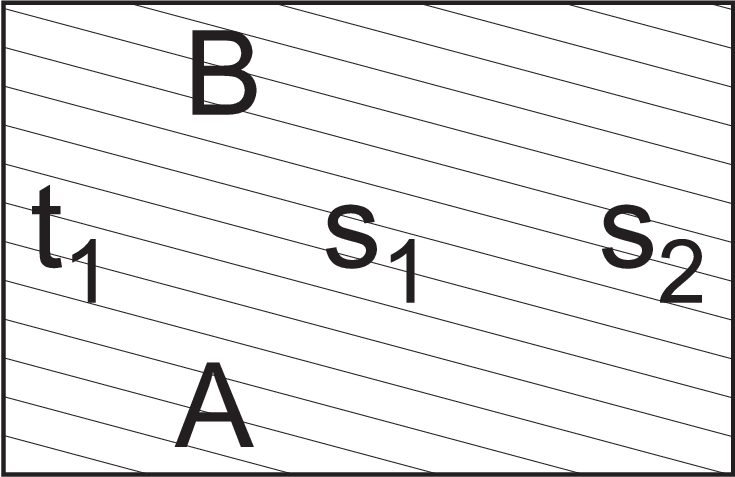}} \qquad 
  \subfloat[$\WW_1$]{\label{fig:part1}\includegraphics[width=2.7cm]{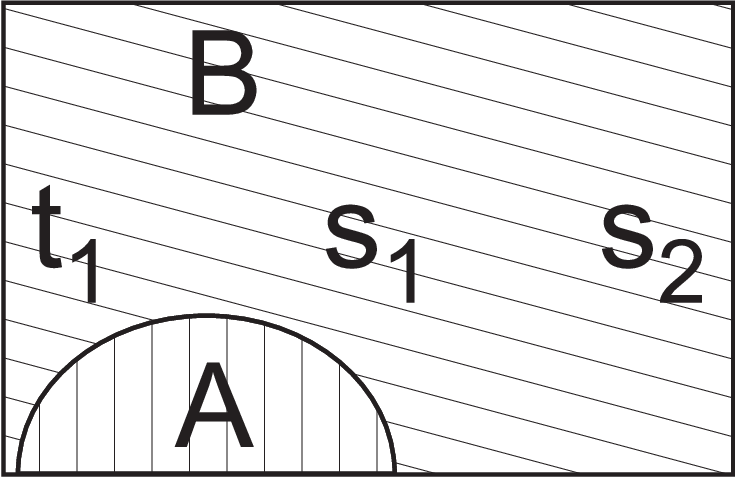}} \qquad              
  \subfloat[$\WW_2$]{\label{fig:part3}\includegraphics[width=2.7cm]{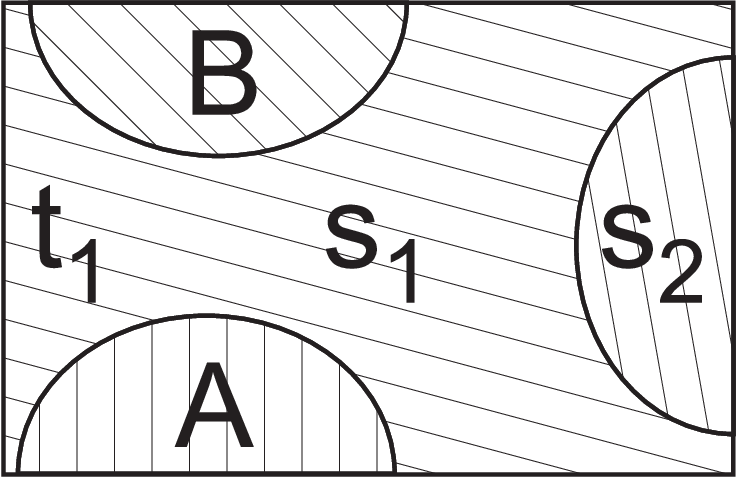}}
  \caption{Example and partitions during algorithm run}
  \label{fig:partitionsbla}
\end{figure}

We assume that $p=q=\frac{1}{2}$, as the pictures are easier to draw in that case, but we would like to stress that the
same arguments work for all other choices (as long as $p$ and $q$ are not equal to $0$ or $1$).

\begin{remark}
In the following graphical representations of the convex sets,
we add dots for every result of a Dirac determinate scheduler (according to \cite{segala:02})
whenever we draw the convex sets of reachable distributions as subsets of $\mathbb{R}^n$. 
Dots that are not extremal points may safely be omitted, as they can be reached as convex combinations of the extremal points.
\end{remark}

{\bf First round:} Start with the partition $\WW_0 = \{\{s_1, s_2, t_1, A, B\} \}$ (cf.~Fig.~\ref{fig:part0}). 
Observe that in the loop from line 9 to line 18 we can never find a vanishing representation of a nn-vanishing state, 
as no state may leave its equivalence class with some probability greater than zero.
Therefore we get $\TANGSTATES=\{s_1, s_2, t_1, A, B\}$.

Now we have to find a splitter with respect to $(\TANGSTATES, \WW_0)$. Suppose that we check the sets $S(\cdot, b)/\WW_0$.
Here we  
see that:
$$S(x,b)/\WW_0=    \begin{cases}
             \includegraphics[width=3cm]{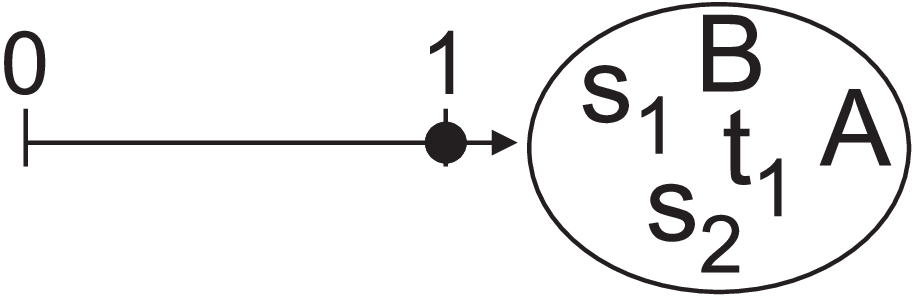} & \text{for }x\in \{B, s_1, s_2, t_1 \} \\
             \emptyset                              & \text{otherwise}
             \end{cases}$$
So we have found a splitter.
Refining according to $(\{s_1, s_2, t_1, A, B\},b,\WW_0)$ leads to $\WW_1 = \{\{s_1, s_2, t_1, B\}, \{A\}\}$ (cf.~Fig.~\ref{fig:part1}).

{\bf Second round:}
We first have to detect the nn-vanishing states with respect to the current partition. We calculate 
$S(x,\tau)$ for every state, verify if it is possible to reach another equivalence class and see whether one single $\tau$ transition
suffices. The values of $S(x,\tau)$ are given in Tab.~\ref{tab:sxtau}.
\begin{table}
\begin{tabular}{|c|c|c|c|c|c|c|c|} \cline{1-2} \cline{4-5} \cline{7-8}
$x$ & $S(x,\tau)/\WW_1$ & \hspace{0.5cm} & $x$ & $S(x,\tau)/\WW_1$ & \hspace{0.5cm} & $x$ & $S(x,\tau)/\WW_1$ \\ \cline{1-2} \cline{4-5} \cline{7-8} \cline{1-2} \cline{4-5} \cline{7-8}
 & \multirow{4}{*}{\includegraphics[width=3cm]{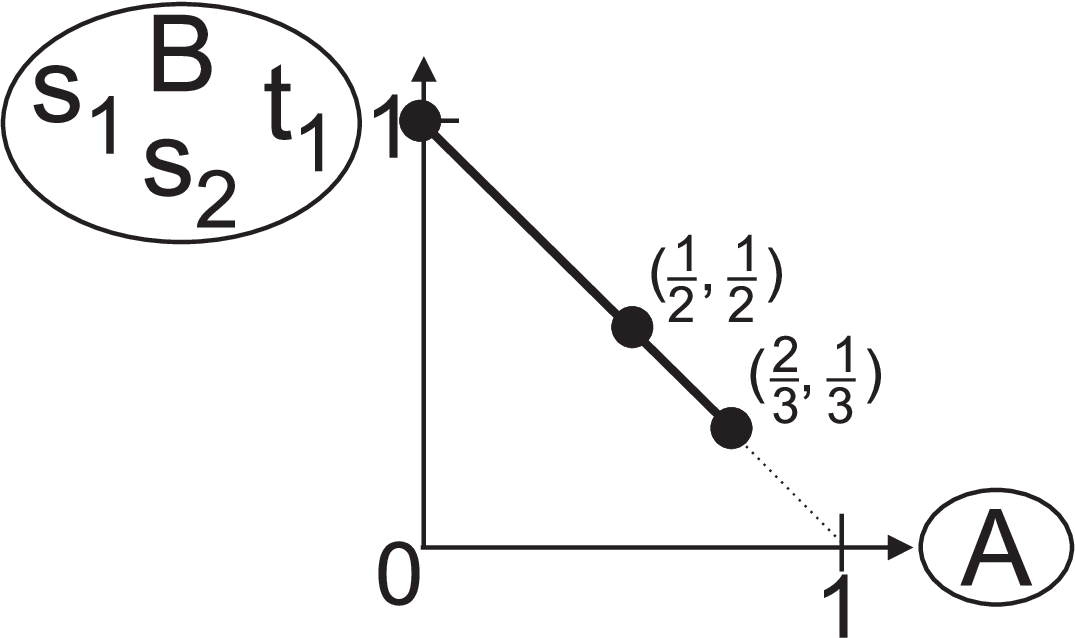}} &  &  & \multirow{4}{*}{\includegraphics[width=3cm]{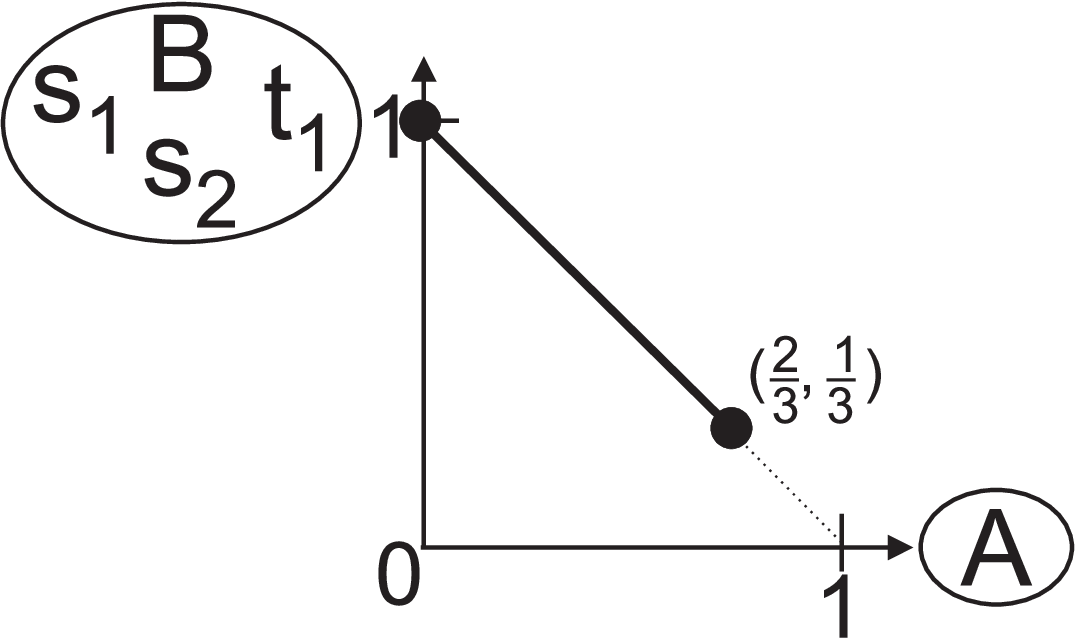}} &  &  & \multirow{4}{*}{\includegraphics[width=3cm]{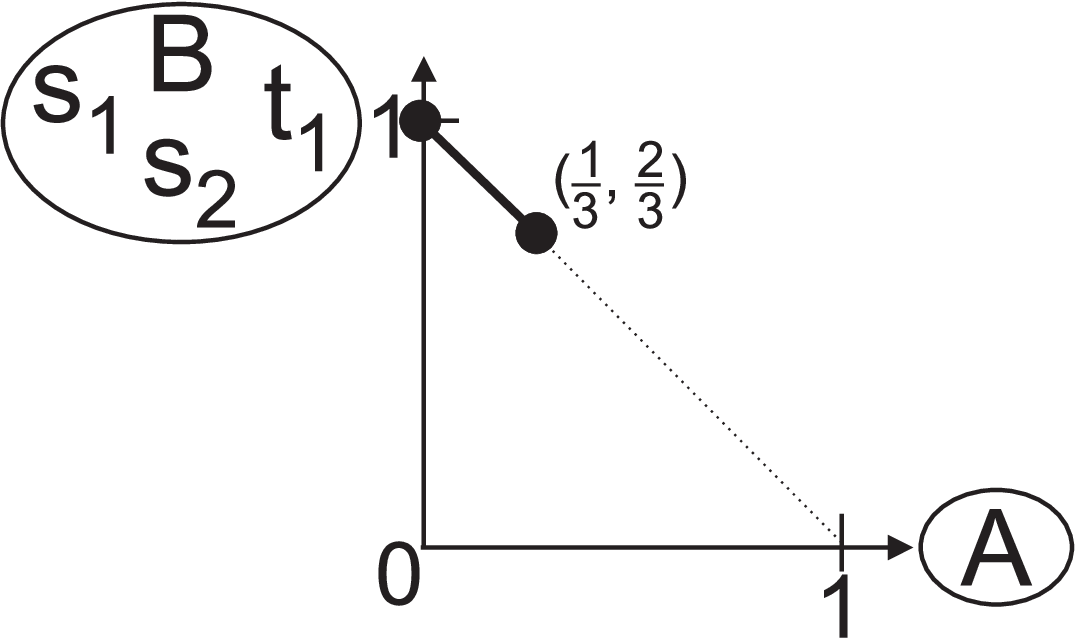}} \\
\qquad &  &  &  &  &  &  &  \\ 
$s_1$ &  &  & $t_1$ &  &  & $s_2$ &  \\ 
\qquad &  &  &  &  &  &  &  \\ 
\qquad &  &  &  &  &  &  &  \\ \cline{1-2} \cline{4-5} \cline{7-8} 
 & \multirow{4}{*}{\includegraphics[width=3cm]{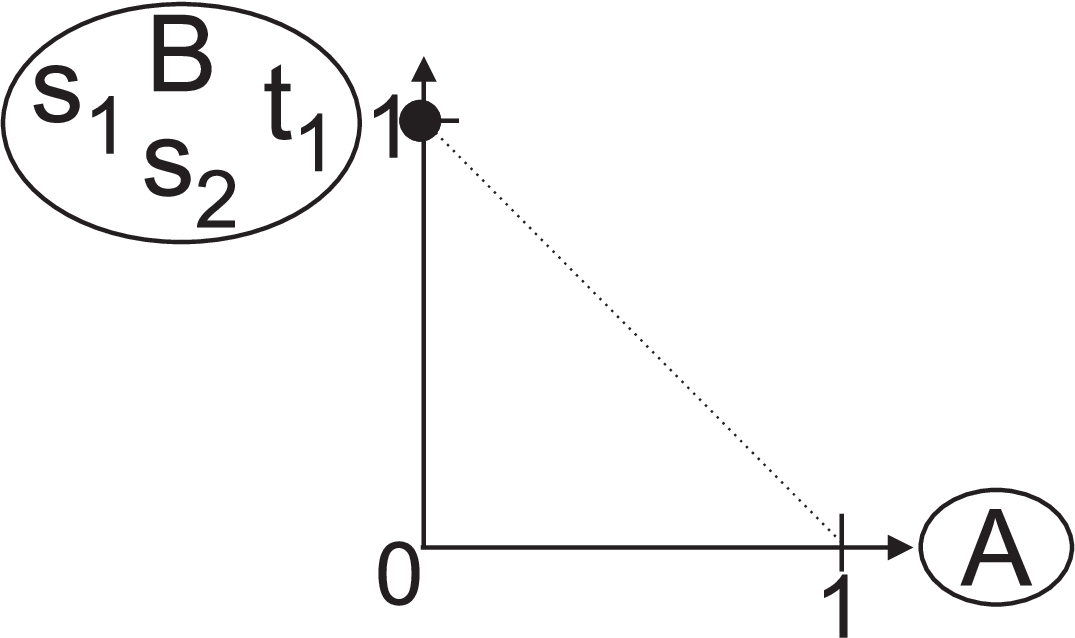}} & &  & \multirow{4}{*}{\includegraphics[width=3cm]{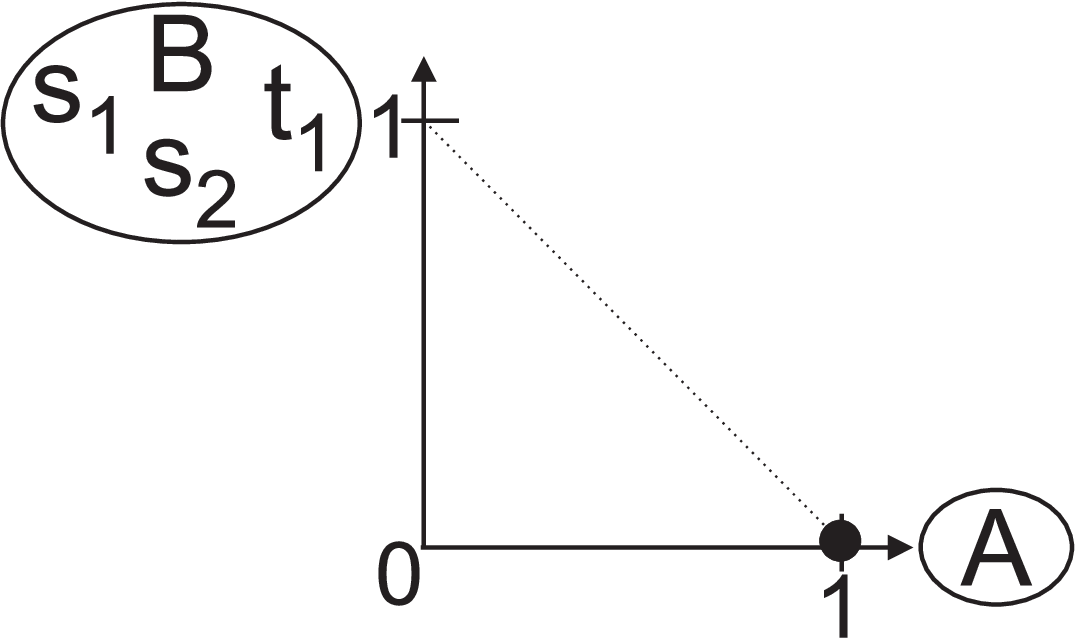}}   \\ 
\qquad &  &  &  &     \\
$B$ &  &  & $A$ &     \\
\qquad &  &  &  &      \\
\qquad &  &  &  &     \\ \cline{1-2} \cline{4-5}
\end{tabular}
\caption{$S(x,\tau)$}
\label{tab:sxtau}
\end{table}
Firstly notice that both $A$ and $B$ cannot be nn-vanishing, as they have no possibility of leaving their equivalence classes.
Notice also, that even if $s_2$ is trivially vanishing, as it has only one single emanating $\tau$ transition, 
we cannot detect it as 
nn-vanishing (the only vanishing representation that leaves the class would be $P_{(s_2,\frac{1}{3}\De_A\oplus \frac{2}{3}\De_B)}$, but 
$S(s_2,b)/\WW_1\neq S_{\frac{1}{3}\De_A\oplus \frac{2}{3}\De_B}(s_2,b)/\WW_1$).
Regarding $s_1$ we see that we cannot omit transition $s_1 \stackrel{\tau}{\rightarrow}{\frac{1}{2}\De_A\oplus \frac{1}{2}\De_{s_2}}$, as 
$S_{(\De_{s_2})}(s_1,\tau)/\WW_1=S(B,\tau)/\WW_1 \neq S(s_1,\tau)$. 
But notice also that $s_1 \stackrel{\tau}{\rightarrow} \De_{s_2}$ cannot be omitted, as $S(s_1,b)/\WW_1=S(B,\tau)/\WW_1$,
but $S_{(\frac{1}{2}\De_A\oplus \frac{1}{2}\De_{s_2})}(s_1,b)/\WW_1=\emptyset$. So we see that $s_1$ cannot be nn-vanishing.
With the same argument we see that also $t_1$ cannot be nn-vanishing. 
Therefore we get $\TANGSTATES=\{s_1, s_2, t_1, A, B\}$.

Now we look for splitters with respect to $(\WW_1, \TANGSTATES)$.
Looking at $C=\{s_1,s_2,t_1,B\}$ we see in routine FindWeakSplit that we can use a splitter $(C,\tau,\WW_1)$ and get the partition
$\WW_2 = \{\{s_1, t_1\}, \{s_2\}, \{A\}, \{B\}\}$ (cf.~Fig.~\ref{fig:part3},
note that $S(s_1,\tau)/\WW_1=S(t_1,\tau)/\WW_1$, as $(\frac{1}{2},\frac{1}{2})$ is \emph{not} a generator of the convex set).

{\bf Third round:} We first have to detect nn-vanishing states. 
It is clear that $s_2$ must be nn-vanishing as it can leave its class and only has a single
outgoing $\tau$ transition. With the same arguments as above we see that both $s_1$ and $t_1$ must be nn-tangible.
So we get $\TANGSTATES=\{s_1, t_1, A, B\}$.

Now we again can look for splitters, but have to consider the restriction to $\TANGSTATES$.
Notice that with coordinates $[s_1]=[t_1]$, $[s_2]$, $[A]$, $[B]$ we have 
$$S(s_1,\tau)/\WW_2 =CHull(\left(\begin{array}{c}1\\0\\0\\0\end{array}\right), \left(\begin{array}{c}0\\1\\0\\0\end{array}\right), \left(\begin{array}{c}0\\0\\0\\1\end{array}\right), \left(\begin{array}{c}0\\\frac{1}{2}\\\frac{1}{2}\\0\end{array}\right), \left(\begin{array}{c}0\\0\\ \frac{2}{3}\\ \frac{1}{3}\end{array}\right)).$$
We want to calculate the restriction $S(s_1,\tau)|_{\TANGSTATES}/\WW_2$.
Let us for the moment ignore the vertex $[s_1]\in S(s_1,\tau)/\WW_2$. Then we get the picture in Fig.~\ref{fig:van} for $S(s_1,\tau)|_{s_1=0}/\WW_2$.
\begin{figure}
  \centering
  \subfloat[$S(s_1,\tau)|_{s_1=0}/\WW_2$]{\label{fig:van}\includegraphics[width=3.5cm]{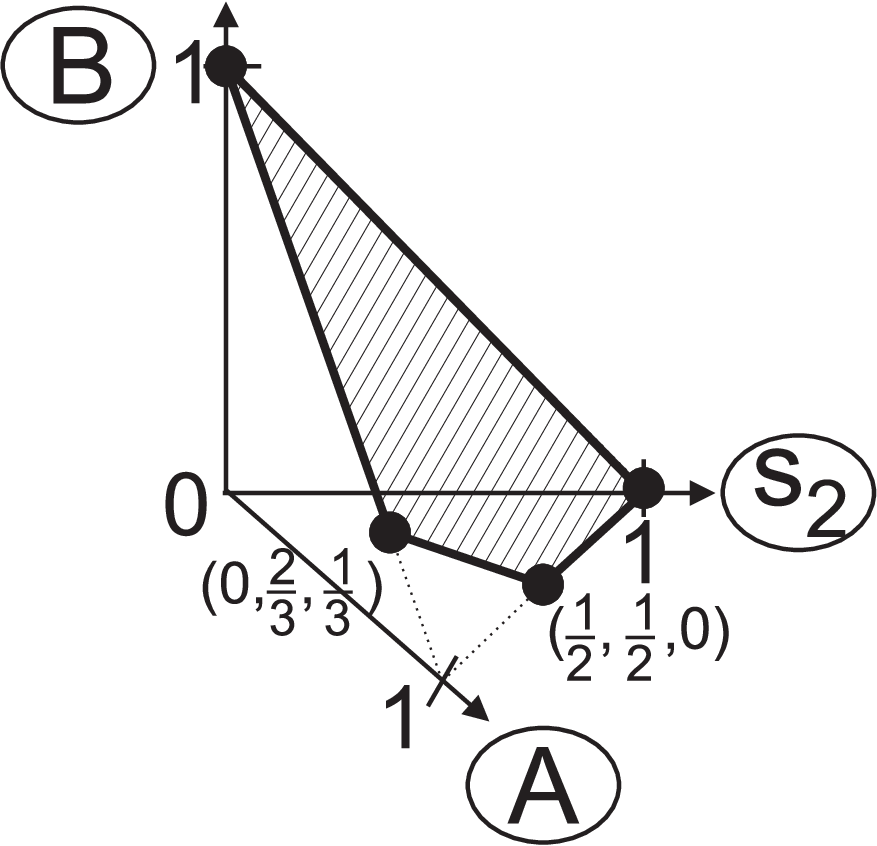}} \qquad
  \subfloat[$S(s_1,\tau)|_{s_1,s_2=0}/\WW_2$]{\label{fig:van3}\includegraphics[width=3.5cm]{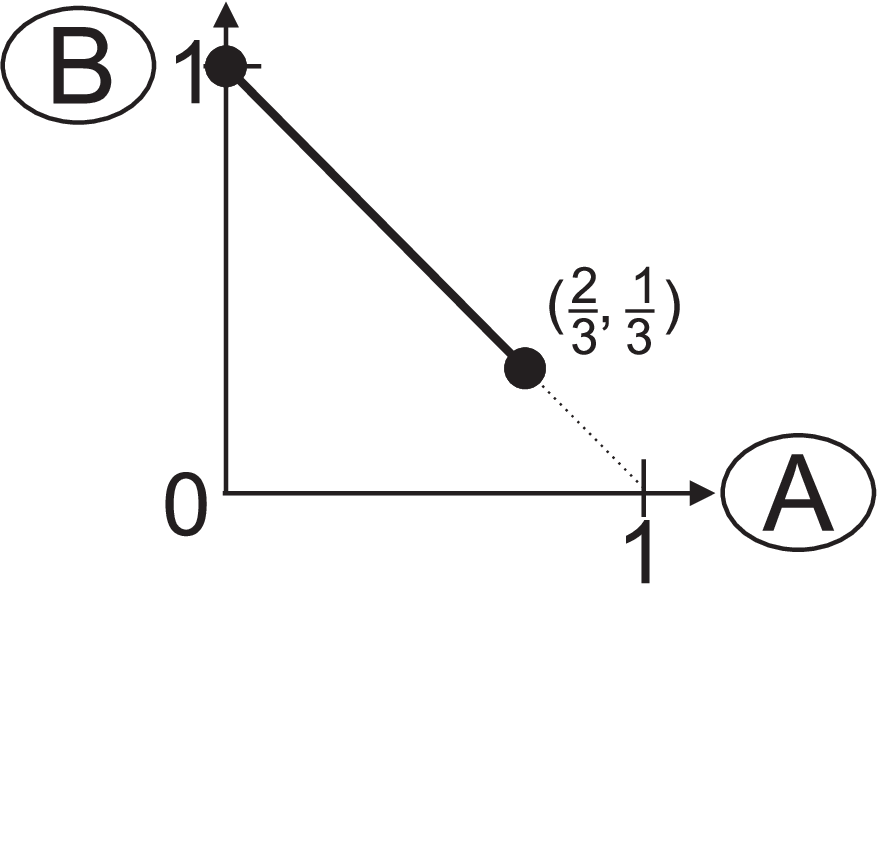}} \qquad 
  \subfloat[$S(s_1,\tau)|_{\TANGSTATES}/\WW_2$]{\label{fig:van2}\includegraphics[width=3.5cm]{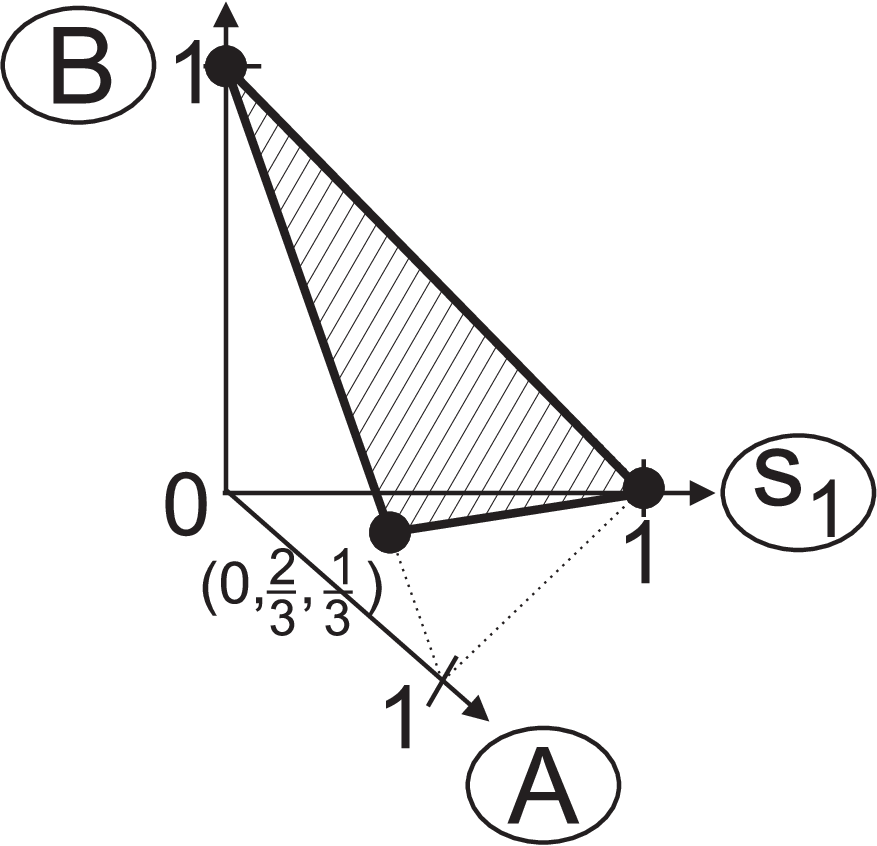}} 
  \caption{Convex sets}
  \label{fig:partitionsbla2}
\end{figure}
We see that the restriction of this set to $\TANGSTATES$ gives only the line from $(0,0,1)$ to $(0,\frac{2}{3},\frac{1}{3})$ (cf.~Fig.~\ref{fig:van3}), therefore we conclude
that $S(s_1,\tau)|_{\TANGSTATES}/\WW_2$ is the set given in Fig.~\ref{fig:van2}. We get the same set for $S(t_1,\tau)|_{\TANGSTATES}/\WW_2$
(here, no nn-vanishing state has to be ignored).
Looking at all other sets $S(\cdot, \cdot)$ we find no other splitter, so $\WW_2$ cannot be refined.

With the partition $\WW_2$ and the set of stable states $\TANGSTATES$ we have reached our fixed point, the algorithm
terminates and we see that $s_1$ and $t_1$ are still in the same partition, so they are weakly bisimilar.

\begin{remark}[Optimisations]
A few optimisations can be performed:
\begin{itemize}
  \item In every round states without other $\tau$ transitions than the loop will always be detected as stable, so this set can be separated as a 
         preprocessing step.
  \item All states with only one single outgoing $\tau$ transitions and without non-$\tau$ transitions can safely be eliminated in advance
         (cf.~Lemma \ref{lemma:1}).
\end{itemize}
\end{remark}

\subsection{Complexity considerations}
\label{subsec:Complexity}
Our algorithm needs to compute the sets $S(s,a)$ and $S_{\nu}(s,a)$ for all $s\in S$, $a\in Act$ and all distributions $\nu$ that can be reached from 
$s$ by a $\tau$ transition driven by a Dirac determinate scheduler. The sets have to be considered with respect to certain partitions.
According to Sec.~7 in \cite{segala:02} the problem of finding one of the above sets is already exponential. Further the Dirac determinate schedulers needed
to find nn-vanishing states are also exponentially many, as pointed out in Example 1 of \cite{segala:02}.
The restriction operation to nn-tangible states is negligible, as the generating points of the above sets, that have non-zero probabilities for nn-tangible states,
can simply be omitted to describe the restricted sets.
%

\section{Relation to Deng-Hennessy bisimulation}
\label{sec:deng-hennessy}
Recently an alternative
distribution-based 
bisimulation $\approx_{bis}$ for Markov Automata has been defined \cite{deng-hennessy:2011}.
One key property of $\mu\approx_{bis}\gamma$ is that whenever $\mu \stackrel{a}{\Rightarrow}_C \oplus_{i\in I}p_i \mu_i$ then also
$\gamma \stackrel{a}{\Rightarrow}_C \oplus_{i\in I}p_i \gamma_i$ where $\mu_i\approx_{bis}\gamma_i$ for all $i\in I$ and vice versa.
This assumption can be directly fed into the proof of our Thm.~\ref{th:main} (instead of having to use Lemma 11 from \cite{avacsreport}).
The reduction to Dirac determinate schedulers works similarly as for the 
bisimulation $\approx$,
as we (as well as \cite{tacas:13}) only use 
standard 
arguments for PA
 which also apply to the Deng-Hennessy setting.
Therefore we conclude that our approach is also capable of deciding $\approx_{bis}$.

\section{Related work}
\label{sec:saarbruecken_small}
Recently, in \cite{qest:13} an alternative approach has been presented to solve the weak bisimulation problem for MA.
We now sum up the analogies and differences, omitting the proofs. 
In the approach of \cite{qest:13}, MEC contractedness plays a crucial role:
\begin{definition}[Maximal End Components, Definitions 6 and 7 in \cite{qest:13}]
Given a PA $P=(S,Act,T,s_0)$, a \emph{maximal
end component} (mec) is a maximal set $C\subseteq S$ such that for each $s, t \in C: s \stackrel{\tau}{\Rightarrow}_C \Delta_t$ and $t \stackrel{\tau}{\Rightarrow}_C \Delta_s$.
A PA $P=(S,Act,T,s_0)$ is called \emph{mec-contracted}, if for each pair of states $(s,t)\in S\times S$ it holds that 
$(s \stackrel{\tau}{\Rightarrow}_C \Delta_t$ and $t \stackrel{\tau}{\Rightarrow}_C \Delta_s) \Rightarrow s=t$.
\end{definition}

\begin{definition}[Behaviourally pivotal state \cite{qest:13}]
We call a state s behaviourally pivotal, if $s\stackrel{\tau}{\rightarrow} \mu$ implies that $s$ and $\mu$ are
not observation equivalent, i.e. $\Delta_s\not \approx \mu$.
It is not behaviourally pivotal if there exists (at least) one transition $s\stackrel{\tau}{\rightarrow} \mu$
such that $\Delta_s\approx \mu$.
\end{definition}

\begin{lemma}[Relating vanishing states to the definitions of \cite{qest:13}]
On MEC-contracted PA, ``vanishing'' corresponds to ``not behaviourally pivotal'' and
``tangible'' corresponds to ``behaviourally pivotal''.
\end{lemma}

Note that MEC-contractedness is crucial for this coincidence. Omitting this precondition, the definitions are different:
Even if ``not behaviourally pivotal'' is defined for arbitrary PA, the definition only makes sense for MEC-contracted PA.
Look at the PA given in Fig.~\ref{fig:beh_pivotal_example}. 
Of course $s$ and $s'$ are in the same class with respect to $\approx_\Delta$.
Therefore both $s$ and $s'$ are not behaviourally pivotal, but it does not make sense to think
about ignoring both of them.
In the context of vanishing states we see that $s$ is tangible while $s'$ is trivially vanishing.
The state $s'$ can be safely ignored, that is: eliminated.
\begin{figure}
  \centering
  \includegraphics[width=4cm]{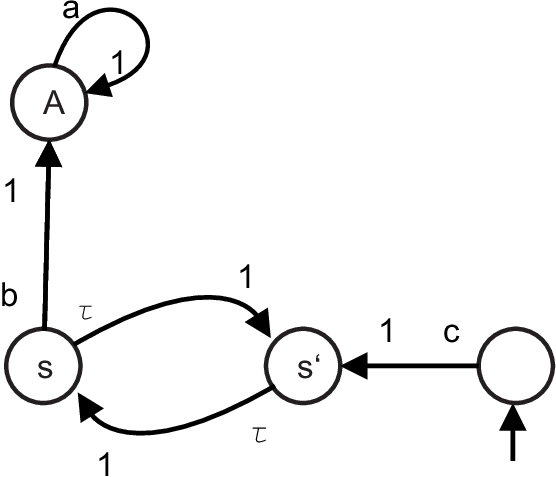}\\
  \caption{Not behaviourally pivotal vs.~vanishing states}
  \label{fig:beh_pivotal_example}
\end{figure}
Note that for MEC-contracted PA this problem doesn't arise.
Our approach has a finer granularity: The set of vanishing states is split into the nn-vanishing and na\"ively vanishing states.
We show in our approach that classes of nn-tangible states (where also na\"ively vanishing states belong to) cannot 
``vanish'', whereas classes of nn-vanishing 
states can ``vanish'' without losing weak bisimilarity.
Lacking this fundamental difference makes the approach of \cite{qest:13} unnecessarily complicated.

Looking at preserving transitions defined in \cite{qest:13}, the picture is similar.
\begin{definition}[Preserving Transitions (adapted from \cite{qest:13})] 
\label{def:preserving}
Let $B$ be an equivalence relation on $S$. A set $P$ of $\tau$-transitions in $T$ is called
preserving with respect to $B$ if for all $(s,\tau,\gamma)\in P$ it holds that
whenever $s\stackrel{a}{\Rightarrow}_C \mu$
then there exist $\mu'$, $\gamma'$
such that $\mu\stackrel{\tau|_P}{\Rightarrow}_C \mu'$ and $\gamma \stackrel{a}{\Rightarrow}_C \gamma'$ and $\mu' \equiv_B \gamma'$.
Here $|_P$ means that we only use transitions from the set $P$.
\end{definition}

\begin{lemma}[Relating vanishing representations to the definitions of \cite{qest:13}]
On a MEC-contracted PA, let $s$ be a vanishing (i.e.~not behaviourally pivotal) state.
Then a ``vanishing representation'' $s\stackrel{\tau}{\rightarrow}{\mu}$ corresponds to a ``preserving transition''.
\end{lemma}
Note that it is important that we use \emph{strong} transitions, as preserving transitions are defined 
as a subset of the set of transitions $T$
(no weak transitions allowed).
It can be shown that every vanishing state has such a 
vanishing representation\footnote{When searching for nn-vanishing states, it is not enough to consider only these strong transitions, as example \ref{ex:nn-van_rep} shows.}.
Similar to the definition of not behaviourally pivotal states, also the definition of preserving transitions only makes sense for MEC-contracted PA:
The restriction to the set $P$ of preserving transitions is not required for the $\tau$ 
parts of the weak $a$ transition from $\gamma$ to $\gamma'$ in Definition \ref{def:preserving}. 
Therefore this definition would render the set of all $\tau$ transitions in 
Fig.~\ref{fig:mec_contractedness} as ``preserving'', but still the $a$ transition 
from $s$ can clearly not be left out, as from $\gamma=\frac{1}{2}\Delta_C\oplus \frac{1}{2}\Delta_D$
still the $a$ transition from $s$ must be used in order to mimick the 
transition to $\frac{1}{2}\Delta_A\oplus \frac{1}{2}\Delta_B$.
\begin{figure}
  \centering
  \includegraphics[width=5cm]{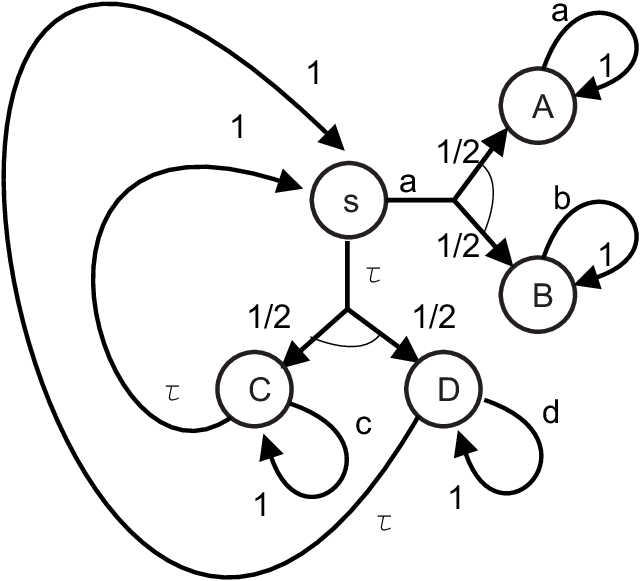}\\
  \caption{A not contracted example}
  \label{fig:mec_contractedness}
\end{figure}
In the context of vanishing states it is clear that there is no vanishing representation for $s$, as then clearly the $a$ transition would get lost.

We collect the following important differences and coincidences from our approach to the approach of \cite{qest:13}:
\begin{compactitem}
  \item The results of \cite{qest:13} only apply to MEC-contracted PA/MA, 
        while our approach does not preassume MEC-contracted MAs, it is a general approach.
  \item The concept of \emph{preserving} transitions --
which remains rather unspecific
 in \cite{qest:13}  and has to be tackled by a brute force attack over all
        possible subsets 
        -- is nicely explained by our concept of \emph{vanishing representations} consisting of strong transitions.
        Especially we have shown that it is enough to consider those sets of preserving transitions where each not behaviourally pivotal state has only 
        one emanating preserving transition (this is a direct consequence of our Lemma \ref{lem:dd-is-enough}).
  \item The definitions of \cite{qest:13} only characterise ``vanishing'' states, but do not distinguish between nn-vanishing and na\"ively vanishing states.
        Therefore that work is lacking the main result similar to our Thm.~\ref{thm:nn-def-reformulation} (nn-vanishing states are the missing part when switching
        from state-based to distribution-based bisimulations) and Thm.~\ref{th:main} (after elimination of nn-vanishing states, a weak bisimulation will be
        a na\"ive weak bisimulation).
  \item With our theory we can easily explain the so-called ``pitfalls'' described in \cite{qest:13}, Example 6 and 7. We show that these are no pitfalls at all in the
        context of nn-vanishing states in \ref{appendix:examples}.
  \item Regarding the complexity, even if not explicitly mentioned in \cite{qest:13} (but as a consequence of the broken ``strong challenger characterisation''), 
        all Dirac determinate schedulers have to be considered 
        for deciding weak bisimilarity between two states. In other words this means that also there the sets $S(s,a)$ are constructed.
        Therefore, the approach of \cite{qest:13} lies in the same complexity class as our approach.
\end{compactitem}

\section{Conclusion}
\label{sec:concl}
We have shown that weak and na\"ive weak bisimulation for MA are closely related by an appropriate formulation of elimination and that
the two notations coincide, when no non-na\"ively vanishing states are present.
We have presented an algorithm for deciding weak MA bisimilarity that, as a by-product, finds non-na\"ively vanishing states and their corresponding 
vanishing representations.
This can also be used to define normal forms for MA. Even with the magnificent results of \cite{turrini:12} it remains an open question whether weak MA bisimulation
can be decided in polynomial time.

\vspace{3ex}
\noindent {\bf Acknowledgements:}
Cordial thanks to Andrea Turrini for giving a beautiful and more readable reformulation of our original
definition of nn-vanishing states \cite{schuster:13} and some interesting discussions on Markov Automata and bisimulations.
We would also like to thank
the anonymous reviewers of Information and Computation
who indicated problems in the proof of the main Theorem, which finally uncovered a problem in Lemma 16 of \cite{avacsreport}
and lead to our new proofs of the main Theorems that are independent of \cite{avacsreport}.

Deutsche Forschungsgemeinschaft
(DFG) supported this work under grant SI 710/7-1,
and we also acknowledge support by the DFG/NWO Bilateral Research Programme ROCKS.







\bibliographystyle{elsarticle-num}
\bibliography{../local}

\appendix
\section{``Continuous'' vs.~``nn-vanishing'' states}
\label{sec:saarbruecken_continuous}

Thm.~2 in \cite{lics:10,avacsreport} relies on Lemma 16 of
\cite{avacsreport}. 
There, the concept of ``continuous'' states is introduced and used for 
proving a key property of weak bisimulation.
However, this appendix points out a counterexample, thus
\cite[Lemma~16]{avacsreport} and therefore Thm.~2 in \cite{avacsreport,lics:10}
have to be considered as yet unproven.
But in this appendix we also show that with the notion of
nn-vanishing states it is possible to prove 
\cite[Lemma~16]{avacsreport}, thus that the lemma 
 and Thm.~2 in \cite{avacsreport,lics:10} are now known to be indeed correct.

\begin{definition}[Continuous state \cite{avacsreport}]
\label{def:continuous}
$P = (S, Act, T, \emptyset, s_0)$ be a PA.
A state $s\in S$ that has a
transition $s\rightarrow \nu$ where $\Delta_s\approx \nu$
, but $\exists t\in Supp(\nu)$ such that $s \not \approx_\Delta t$ 
is called a \emph{continuous} state.
\end{definition}

In order to compare the concepts of ``continuous'' and ``nn-vanishing'' states,
it is convenient to have an alternative characterisation 
of nn-vanishing states.
By Thm.~\ref{thm:nn-def-reformulation} and Lemma \ref{lem:dd-is-enough} we see that we could alternatively define nn-vanishing states
in the following way:

\begin{definition}[nn-vanishing state -- alternative definition to Definition \ref{def:vanishing}]
\label{def:nn-van-alternative}
$P = (S, Act, T, \emptyset, s_0)$ be a PA.
A state $s\in S$ that has a (non-combined!) weak
transition $s\Rightarrow \nu$ where $\Delta_s\approx \nu$
but $\exists t\in Supp(\nu)$ such that $s \not \approx_\Delta t$ 
is called a \emph{nn-vanishing} state.
\end{definition}

So we see that the set of continuous 
states is in general smaller than the set of nn-vanishing states 
(strong transition $s\stackrel{\tau}{\rightarrow}\nu$ vs.~weak transition $s\stackrel{\tau}{\Rightarrow}\nu$). 
With this knowledge, we can give a simple example that renders the proof of \cite[Lemma~16]{avacsreport}
wrong.

\begin{example}[Counterexample: Weak transitions not considered]
The proof of \cite[Lemma~16]{avacsreport} consists of two steps. The first step constructs canonical transitions that resolve ``continuous'' states 
$\mu\Rightarrow \mu^\ast$ and $\gamma \Rightarrow \gamma^\ast$ (where $\mu^\ast \approx \mu$ and $\gamma^\ast \approx \gamma$)
using Dirac determinate schedulers.
In the second step it is shown that every further transition $\mu^\ast \Rightarrow_C \mu'$ (where $\mu'\approx \mu^\ast$) and 
analogously $\gamma^\ast\Rightarrow \gamma'$ (where $\gamma'\approx \gamma^\ast$)
do not change the equivalence classes.
Assume we are given the automata in Fig.~\ref{fig:cont_van_1}. The states $s$, $s'$ and $t$ are nn-vanishing. 
Assume that $\mu=\Delta_s$, $\gamma=\Delta_t$. According to \cite{avacsreport} we see that 
$t$ and $s'$ are continuous while $s$ is not. By the 
construction from \cite{avacsreport} we would get then $\mu^\ast=\mu$ (no state in $\mu$ can change its
equivalence class with a \emph{strong} transition) and $\gamma^\ast=\frac{1}{2}\Delta_A\oplus \frac{1}{2}\Delta_B$.
But now it is trivially wrong that $\mu^\ast$ and $\gamma^\ast$ coincide on classes, as $s$ is clearly nn-vanishing while $A$ and $B$ are not. Therefore
the proof of \cite[Lemma~16]{avacsreport} is incorrect.

\begin{figure}
  \centering
  \subfloat[$s$ and $s'$ nn-vanishing]{\label{fig:E_triv_van2}\includegraphics[height=3cm]{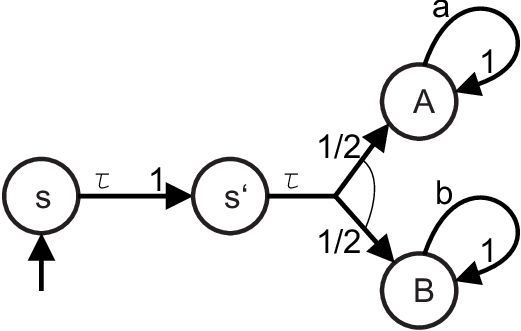}} \qquad
  \subfloat[$t$ nn-vanishing]{\label{fig:E_van2}\includegraphics[height=3cm]{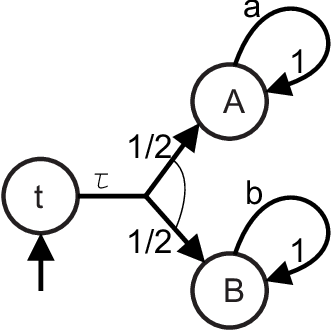}} \\
  \caption{Examples of vanishing states}
  \label{fig:cont_van_1}
\end{figure}
\end{example}

Still, \cite[Lemma~16]{avacsreport} remains correct, and its correctness can be proven
in the nn-vanishing context:
We can substitute every nn-vanishing state 
by its canonical vanishing representation consisting of nn-tangible
states, and the resulting distributions are na\"ively weakly bisimilar
(with respect to the lifting of $\approx_{\text{na\"ive}}$ to distributions).
 This is, with the help of our notion of nn-vanishing states,
now proven by Lemma \ref{lemma:non-vanishing} and Thm.~\ref{th:main}.
The fact that non-combined transitions can be used in this lemma is now proven by our Lemma \ref{lem:dd-is-enough}.



\section{Examples from \cite{qest:13}}
\label{appendix:examples}

The following two examples from \cite{qest:13} are called ``pitfalls'' there.
We show why -- in the context of the nn-vanishing state concept -- these pitfalls are no pitfalls at all.

\subsection{Strong challenger characterisation}
\begin{figure}
  \centering
  \includegraphics[width=7.5cm]{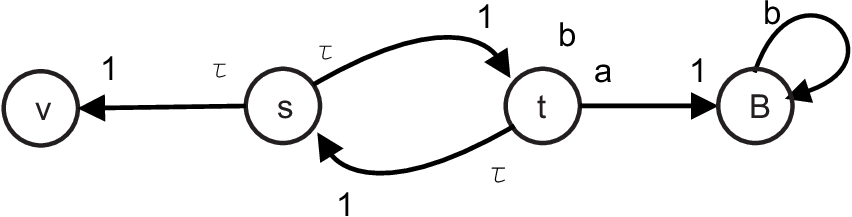}\\
  \caption{Example 6 from \cite{qest:13} -- $s$ and $t$ in one class}
  \label{fig:naively_but_not_mentioned}
\end{figure}
In example 6 of \cite{qest:13} (cf.~Fig.~\ref{fig:naively_but_not_mentioned}) 
it is clear that the preserving-approach fails as long as $v$ remains absorbing:
States $s$ and $t$ belong to one class with respect to $\approx_\Delta$. 
It can be easily verified that both states cannot be vanishing, as no vanishing representation can be found.
As both states are tangible, they actually don't need a special treatment in our algorithm.
Note that ``preserving transitions'' in our understanding are only necessary for nn-vanishing states
and \emph{not} for all vanishing states in order to solve the decision problem. 

\subsection{Brute force attack for ``preserving'' transitions}
The problem of example 7 of \cite{qest:13},
as exemplified by Fig.~\ref{fig:bulls_eye}, doesn't hit the bull's eye.
The basic question is not ``which transitions can be omitted?'' (or alternatively ``which transitions are preserving?''). 
The first question must rather be
``are states $s$ and $t$ nn-vanishing or not?''.
If they are nn-tangible, not any transition may be omitted.
If they are nn-vanishing, Lemma~\ref{lem:dd-is-enough} justifies that for both states the \emph{same} transition must be omitted 
(as long as the successor distributions are \emph{not} bisimilar), 
as Fig.~\ref{fig:bulls_eye}
shows (assume that the ``triangle'' (``pentagon'') distribution from \cite{qest:13} corresponds to 
state C (D) in our example. Obviously the automaton is MEC-contracted. Clearly 
states $A$, $B$ and $C$ are not weakly bisimilar. Assume that $p\in (0,1)$. Then D is trivially nn-vanishing.
Further it is clear that the $\tau$ transitions from $s$ and $t$ to $C$ can be omitted, as they may be weakly mimicked.
But then (and only then) we have vanishing representations of states $s$ and $t$.
So we conclude that the statement in \cite{qest:13} that ``Then, clearly, none of the transitions is
preserving'' is in general wrong. It rather depends on the context whether $s$ or $t$ are nn-vanishing or not.
\begin{figure}
  \centering
  \includegraphics[width=4.5cm]{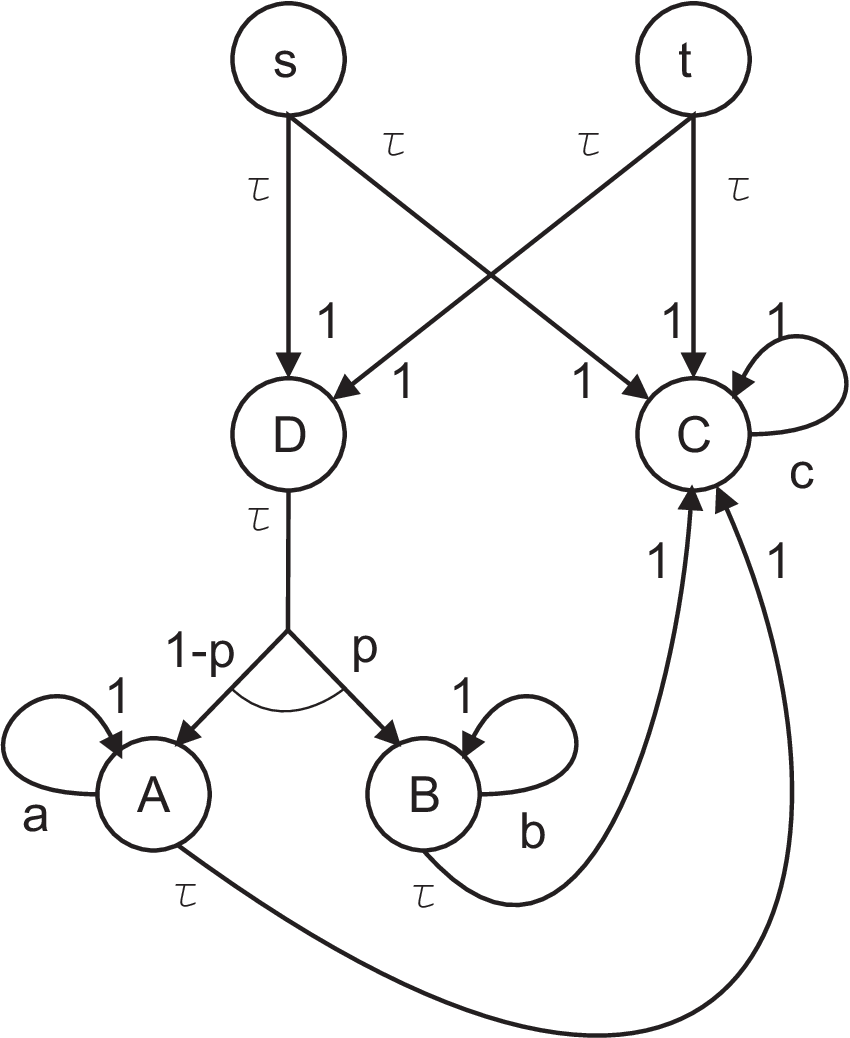}\\
  \caption{An instance of Example 7 from \cite{qest:13} -- nn-vanishing case
}
  \label{fig:bulls_eye}
\end{figure}







\end{document}